%% file: note_seriation_nips.tex
\documentclass{article}
\usepackage{color, xcolor}

\usepackage{amsmath}
\usepackage{amsthm}

\usepackage{amssymb}
\usepackage{bbm}
\usepackage{comment}

\usepackage[preprint]{nips_2018}

\usepackage[utf8]{inputenc} 
\usepackage[T1]{fontenc}    
\usepackage{hyperref}       
\usepackage{url}            
\usepackage{booktabs}       
\usepackage{amsfonts}       
\usepackage{nicefrac}       
\usepackage{microtype}      

\def\yy{{\boldsymbol y}}

\usepackage{graphicx}
\usepackage{subcaption}
\usepackage{multirow}
\usepackage{wrapfig}
\usepackage{multicol}

\input{defs.tex}

\input{defsperms.tex}

\usepackage{algorithmic}

\title{
Reconstructing Latent Orderings\\ by Spectral Clustering
}

\author{ Antoine Recanati\thanks{equal contribution}~~$^{1,2,3,4}$,  Thomas Kerdreux \footnotemark[1]~~$^{1,2,3 }$  and  Alexandre d'Aspremont $^{1,2,3,4}$\\
~\\
$^1$Département d’informatique de l’\'Ecole normale supérieure\\
$^2$PSL Research University, 75005 Paris, France\\
$^3$INRIA ; $^4$CNRS
}

\begin{document}
\maketitle

\begin{abstract}
Spectral clustering uses a graph Laplacian spectral embedding to enhance the cluster structure of some data sets.
When the embedding is one dimensional, it can be used to sort the items (spectral ordering).
A number of empirical results also suggests that a multidimensional Laplacian embedding enhances the latent ordering of the data, if any. This also extends to circular orderings, a case where unidimensional embeddings fail. We tackle the task of retrieving linear and circular orderings in a unifying framework, and show how a latent ordering on the data translates into a filamentary structure on the Laplacian embedding.
We propose a method to recover it, illustrated with numerical experiments on
synthetic data and real DNA sequencing data.
The code and experiments are available at \url{https://github.com/antrec/mdso}.
\end{abstract}

\input{sections/1_Introduction.tex}
\input{sections/2_Laplacian_Embedding_circular.tex}
\input{sections/3_Recovery_algorithm.tex}
\input{sections/4_numerical_results.tex}

\section*{Acknowledgements}
AA is at the d\'epartement d'informatique de l'ENS, \'Ecole normale sup\'erieure, UMR CNRS 8548, PSL Research University, 75005 Paris, France, and INRIA Sierra project-team. The authors would like to acknowledge support
from the \textit{data science} joint research initiative with the \textit{fonds AXA pour la recherche} and Kamet Ventures. TK acknowledges funding from the CFM-ENS chaire {\em les mod\`eles et sciences des donn\'ees}. Finally the authors would like to thanks Raphael Berthier for fruitfull discussions.

\bibliographystyle{apalike}
\bibliography{biblio}{}

\clearpage
\appendix
\input{sections/8_appendix_algorithms.tex}
\input{sections/10_appendix_supp_figure.tex}

\input{sections/5_appendix_circular_proof.tex}
\input{sections/6_appendix_perturbation_analysis.tex}

\end{document}

%% file: defs.tex
\newtheorem{theorem}{Theorem}[section]
\newtheorem{proposition}[theorem]{Proposition}
\newtheorem{definition}[theorem]{Definition}
\newtheorem{lemma}[theorem]{Lemma}

\renewenvironment{proof}{\textbf{Proof.}}{\QED\bigskip}

\usepackage{algorithmic,algorithm}

\definecolor{ddarkbrown}{rgb}{0.5,0.2,0.05} \definecolor{bbluegray}{rgb}{0.05,0,0.5}

\newcommand{\BEAS}{\begin{eqnarray*}}
\newcommand{\EEAS}{\end{eqnarray*}}
\newcommand{\BEA}{\begin{eqnarray}}
\newcommand{\EEA}{\end{eqnarray}}
\newcommand{\BEQ}{\begin{equation}}
\newcommand{\EEQ}{\end{equation}}
\newcommand{\BIT}{\begin{itemize}}
\newcommand{\EIT}{\end{itemize}}
\newcommand{\BNUM}{\begin{enumerate}}
\newcommand{\ENUM}{\end{enumerate}}

\newcommand{\BA}{\begin{array}}
\newcommand{\EA}{\end{array}}

\newcommand{\eg}{{\it e.g.}}
\newcommand{\ie}{{\it i.e.}}
\newcommand{\etc}{{\it etc.}}
\newcommand{\ones}{\mathbf 1}

\newcommand{\reals}{{\mathbb R}}

\newcommand{\symm}{{\mbox{\bf S}}}  

\newcommand{\Tr}{\mathop{\bf Tr}}
\newcommand{\diag}{\mathop{\bf diag}}

\newcommand{\idm}{\mathbf{I}}

\newcommand{\QED}{~~\rule[-1pt]{6pt}{6pt}}

\newcommand{\argmax}{\mathop{\rm argmax}}

\renewcommand\Re{\operatorname{Re}}

\def\uu{{\boldsymbol u}}

\def\xx{{\boldsymbol x}}

%% file: defsperms.tex
\providecommand{\cP}{\mathcal{P}}

\providecommand{\PAP}{\Pi A \Pi^T}
\providecommand{\cR}{\mathcal{L}_{R}}
\providecommand{\circR}{\mathcal{C}_{R}}

\providecommand{\cH}{\mathcal{H}}
\providecommand{\Lrw}{L^{\text{rw}}}
\providecommand{\Lsym}{L^{\text{sym}}}
\providecommand{\kLE}[1]{#1-$\mathbf{LE}$}
\providecommand{\SCR}{\mathcal{C}_{R}^*}

\providecommand{\st}{\text{such that}}
\providecommand{\find}{\text{find}}

%% file: sections/1_Introduction.tex
\section{Introduction}

The seriation problem seeks to recover a latent ordering from similarity information. We typically observe a matrix measuring pairwise similarity between a set of $n$ elements and assume they have a serial structure, \ie~they can be ordered along a chain where the similarity between elements decreases with their distance within this chain. In practice, we observe a random permutation of this similarity matrix, where the elements are not indexed according to that latent ordering. Seriation then seeks to find that global latent ordering using only (local) pairwise similarity.

Seriation was introduced in archaeology to find the chronological order of a set of graves. Each contained artifacts, assumed to be specific to a given time period. The number of common artifacts between two graves define their similarity, resulting in a chronological ordering where two contiguous graves belong to a same time period.
It also has applications in, \eg, envelope reduction \citep{Barn95}, bioinformatics \citep{atkins1996physical,cheema2010thread,jones2012anges} and DNA sequencing \citep{Meid98,Garr11,recanati2016spectral}.

In some applications, the latent ordering is circular. For instance, in {\it de novo} genome assembly of bacteria, one has to reorder DNA fragments subsampled from a circular genome.

In biology, a cell evolves according to a cycle: a newborn cell passes through diverse states (growth, DNA-replication, \etc) before dividing itself into two newborn cells, hence closing the loop. Problems of interest then involve collecting cycle-dependent data on a population of cells at various, unknown stages of the cell-cycle, and trying to order the cells according to their cell-cycle stage.
Such data include gene-expression \citep{liu2017reconstructing}, or DNA 3D conformation data \citep{liu2018unsupervised}.
In planar tomographic reconstruction, the shape of an object is inferred from projections taken at unknown angles between 0 and $2\pi$. Reordering the angles then enables to perform the tomography \citep{coifman2008graph}.

The main structural hypothesis on similarity matrices related to seriation is the concept of $R$-matrix, which we introduce below, together with its circular counterpart.
\begin{definition}\label{def:R-mat}
	We say that $A\in\symm_n$ is a R-matrix (or Robinson matrix) iff it is symmetric and satisfies
		$A_{i,j} \leq  A_{i,j+1}$ and $A_{i+1,j} \leq A_{i,j}$ in the lower triangle, where $1\leq j < i \leq n$.
\end{definition}
\begin{definition}\label{def:circ-R-mat}
	We say that $A\in\symm_n$ is a circular R-matrix iff it is symmetric and satisfies, for all $i \in [n]$,
	$\left(A_{ij}\right)_{j=1}^{i}$ and $\left(A_{ij}\right)_{i=j}^{n}$ are unimodal : they are decrease to a minimum and then increase.
\end{definition}
Here $\symm_n$ is the set of real symmetric matrices of dimension~$n$.
Definition~\ref{def:R-mat} states that when moving away from the diagonal in a given row or column of $A$, the entries are non-increasing, whereas
in Def~\ref{def:circ-R-mat}, the non-increase is followed by a non-decrease. For instance, the proximity matrix of points embedded on a circle follows Def~\ref{def:circ-R-mat}.
Figure~\ref{fig:seriation} displays examples of such matrices.
\begin{figure}[hbt]
	\begin{center}
		\begin{subfigure}[htb]{0.3\textwidth}
			\includegraphics[width=\textwidth]{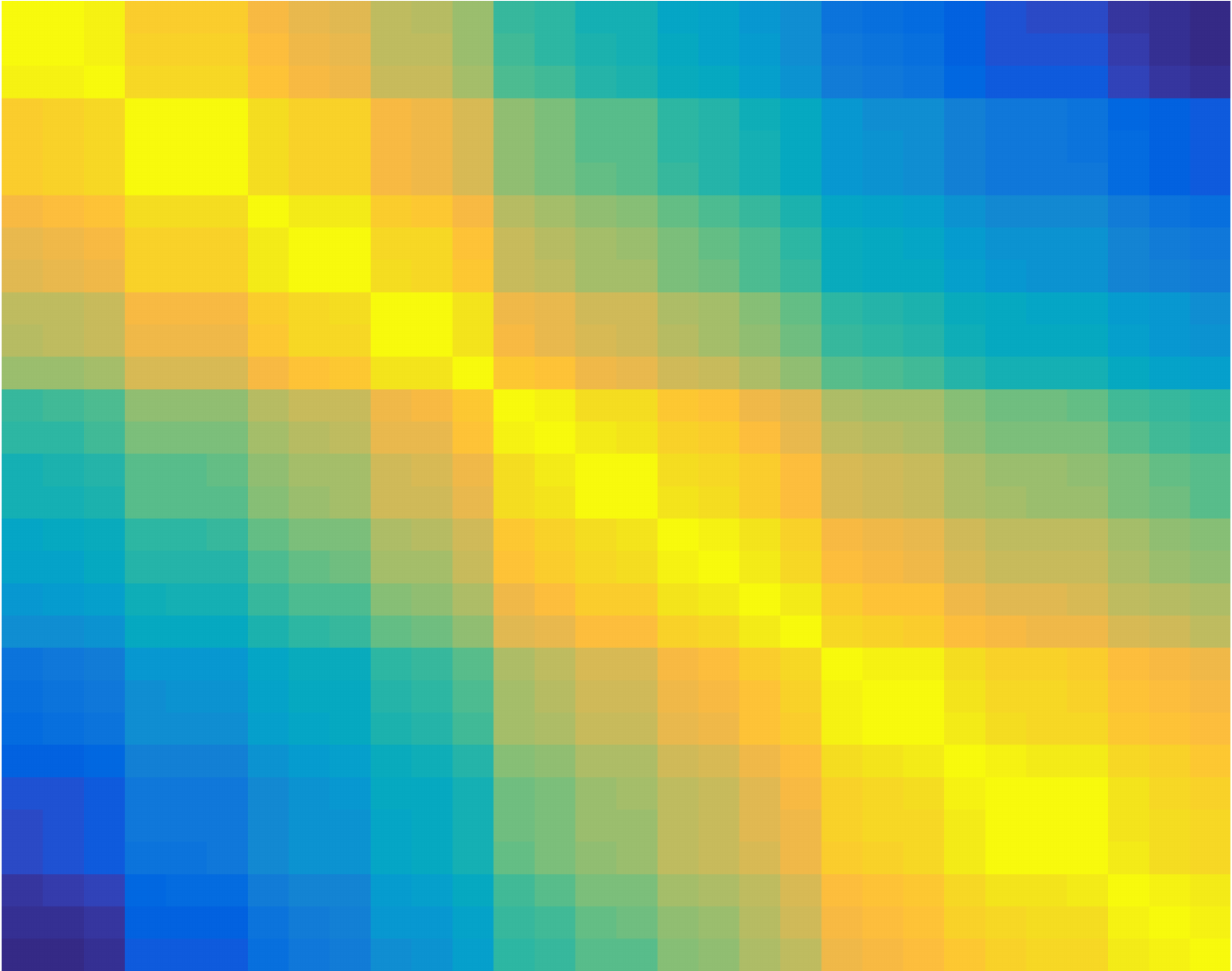}
			\caption{R-matrix}\label{subfig:Rmat}
		\end{subfigure}
		\begin{subfigure}[htb]{0.3\textwidth}
			\includegraphics[width=\textwidth]{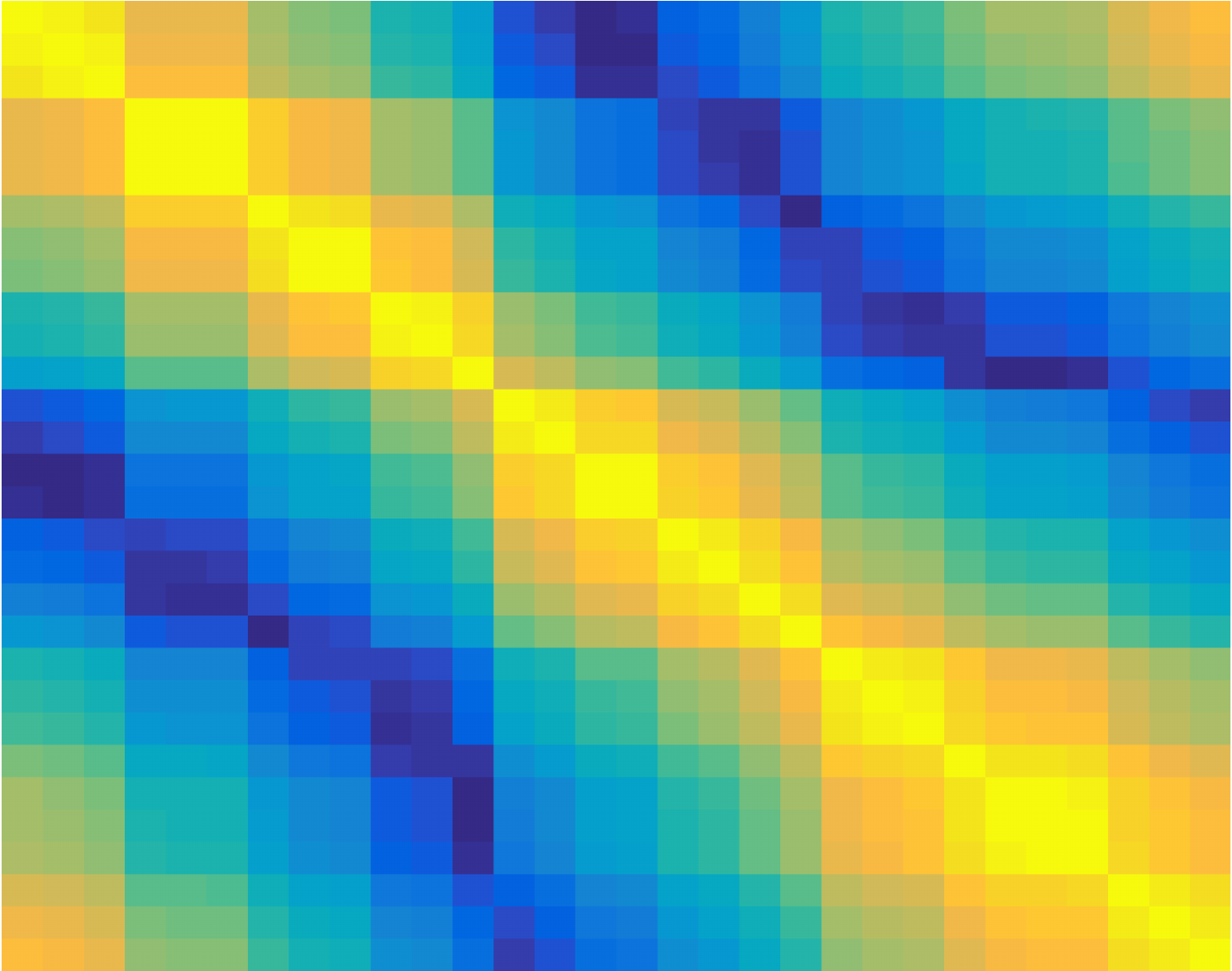}
			\caption{circular R-matrix}\label{subfig:circRmat}
		\end{subfigure}
		\begin{subfigure}[htb]{0.3\textwidth}
			\includegraphics[width=\textwidth]{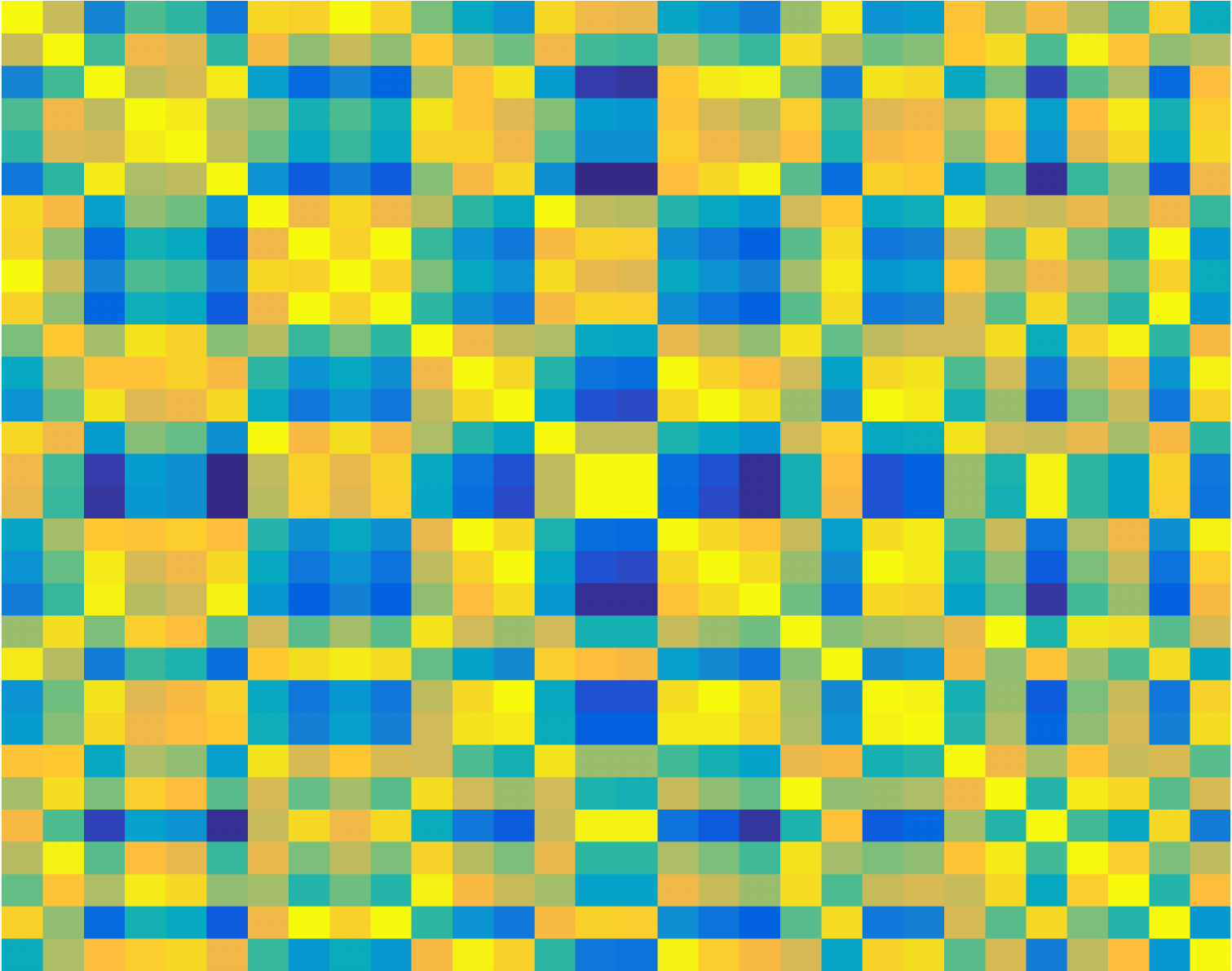}
			\caption{permuted R-matrix}\label{subfig:2DRperm}
		\end{subfigure}
		\caption{
		From left to right, R-matrix (\ref{subfig:Rmat}), circular R-matrix (\ref{subfig:circRmat}), and a randomly permuted observation of a R-matrix (\ref{subfig:2DRperm}). Seriation seeks to recover (\ref{subfig:Rmat}) from its permuted observation (\ref{subfig:2DRperm}).
		}
		\label{fig:seriation}
	\end{center}
	\vskip -0.2in
\end{figure}

In what follows, we write $\cR^n$ (resp., $\circR^n$) the set of R (resp., circular-R) matrices of size $n$, and $\cP_n$ the set of permutations of $n$ elements. A permutation can be represented by a vector $\pi$ (lower case) or a matrix $\Pi \in \{0,1\}^{n \times n}$ (upper case) defined by $\Pi_{ij} = 1$ iff $\pi(i) = j$, and $\pi = \Pi \pi_{Id}$ where $\pi_{Id} = (1, \dots, n)^T$. We refer to both representations by $\cP_n$ and may omit the subscript $n$ whenever the dimension is clear from the context. 
We say that $A \in \symm_n$ is pre-$\cR$ (resp., pre-$\circR$) if there exists a permutation $\Pi \in \cP$ such that the matrix $\PAP$ (whose entry $(i,j)$ is $A_{\pi(i),\pi(j)}$) is in $\cR$ (resp., $\circR$). Given such $A$, Seriation seeks to recover this permutation $\Pi$,
\begin{align}
\BA{lllll}
\find &  \Pi \in \cP & \st & \PAP  \in \cR & \tag{Linear Seriation}\label{eqn:seriation}
\EA
\\
\BA{lllll}
\find  & \Pi \in \cP  & \st & \PAP  \in \circR & \tag{Circular Seriation}\label{eqn:circ-seriation}
\EA
\end{align}
A widely used method for \ref{eqn:seriation} is a spectral relaxation based on the graph Laplacian of the similarity matrix.
It transposes Spectral Clustering \citep{von2007tutorial} to the case where we wish to infer a latent ordering rather than a latent clustering on the data.
Roughly speaking, both methods embed the elements on a line and associate a coordinate $f_i \in \reals$ to each element $i \in [n]$. Spectral clustering addresses a graph-cut problem by grouping these coordinates into two clusters. Spectral ordering \citep{Atkins} addresses \ref{eqn:seriation} by sorting the $f_i$.

Most Spectral Clustering algorithms actually use a Laplacian embedding of dimension $d>1$, denoted \kLE{d} in the following. Latent cluster structure is assumed to be enhanced in the \kLE{d}, and the k-means algorithm \citep{macqueen1967some,hastie2009unsupervised} seamlessly identifies the clusters from the embedding.
In contrast, Spectral Ordering is restricted to $d=1$ by the sorting step (there is no total order relation on $\reals^d$ for $d>1$).
Still, the latent linear structure may emerge from the \kLE{d}, if the points are distributed along a curve.
Also, for $d=2$, it may capture the circular structure of the data and allow for solving \ref{eqn:circ-seriation}.
One must then recover a (circular) ordering of points lying in a $1D$ manifold (a curve, or filament) embedded in $\reals^d$.

In Section~\ref{sec:related-work}, we review the Spectral Ordering algorithm and the Laplacian Embedding used in Spectral Clustering. We mention graph-walk perspectives on this embedding and how this relates to dimensionality reduction techniques. Finally, we recall how these perspectives relate the discrete Laplacian to continuous Laplacian operators, providing insights about the curve structure of the Laplacian embedding through the spectrum of the limit operators. These asymptotic results were used to infer circular orderings in a tomography application in e.g. \citet{coifman2008graph}.
In Section~\ref{sec:theory}, we evidence the filamentary structure of the Laplacian Embedding, and provide theoretical guarantees about the Laplacian Embedding based method for \ref{eqn:circ-seriation}.
We then propose a method in Section~\ref{sec:results} to leverage the multidimensional Laplacian embedding in the context of \ref{eqn:seriation} and \ref{eqn:circ-seriation}.
We eventually present numerical experiments to illustrate how the spectral method gains in robustness by using a multidimensional Laplacian embedding.

\section{Related Work}\label{sec:related-work}

\subsection{Spectral Ordering for Linear Seriation}

\ref{eqn:seriation} can be addressed with a spectral relaxation of the following combinatorial problem,
\begin{align}
\tag{2-SUM} \label{eqn:2sum}
\BA{llll}
\mbox{minimize}  & \sum_{i,j=1}^n 
A_{ij} |\pi_i - \pi_j|^2 & \st & \pi \in \cP_n\\
\EA
\end{align}
Intuitively, the optimal permutation compensates high $A_{ij}$ values with small $|\pi_i - \pi_j|^2$, thus laying similar elements nearby.
For any $f = \left(f(1),\ldots,f(n)\right)^T \in \reals^n$, the objective of \ref{eqn:2sum} can be written as a quadratic (with simple algebra using the symmetry of $A$, see \citet{von2007tutorial}),
\begin{align}\label{eqn:2sum-is-quadratic}
    {\textstyle \sum_{i,j=1}^n}  A_{ij} |f(i) - f(j)|^2 = f^T L_A f
\end{align}
where $L_A \triangleq \diag(A\ones)-A$ is the graph-Laplacian of $A$.
From~\eqref{eqn:2sum-is-quadratic}, $L_A$ is positive-semi-definite for $A$ having non-negative entries, and $\ones = (1, \ldots, 1)^T$ is an eigenvector associated to $\lambda_0 = 0$.

The spectral method drops the constraint $\pi \in \cP_n$ in \ref{eqn:2sum} and enforces only norm and orthogonality constraints, $\|\pi\|=1$, $\pi^T \ones = 0$, to avoid the trivial solutions $\pi = 0$ and $\pi \propto \ones$, yielding,
\begin{align}
\tag{Relax. 2-SUM}\label{eqn:2sum-relaxed}
\BA{llll}
\mbox{minimize}  & f^T L_A f & \st & \| f \|_2 = 1 \:,\: f^T \ones = 0.\\
\EA
\end{align}

This is an eigenvalue problem on $L_A$ solved by $f_{(1)}$, the eigenvector associated to $\lambda_1 \geq 0$ the second smallest eigenvalue of $L_A$. If the graph defined by $A$ is connected (which we assume further) then $\lambda_1 > 0$.
From $f_{(1)}$, one can recover a permutation by sorting its entries.
The spectral relaxation of \ref{eqn:2sum} is summarized in Algorithm~\ref{alg:spectral}.
For pre-$\cR$ matrices, \ref{eqn:seriation} is equivalent to \ref{eqn:2sum} \citep{Fogel}, and can be solved with Algorithm~\ref{alg:spectral} \citep{Atkins}, as stated in Theorem~\ref{thm:spectral-solves-seriation-preR}.
\begin{algorithm}[H]
	\footnotesize
	\caption{Spectral ordering \citep{Atkins}}\label{alg:spectral}
	\begin{algorithmic} [1]
		\REQUIRE Connected similarity matrix $A \in \mathbb{R}^{n \times n}$
		\STATE Compute Laplacian $L_A=\diag(A\ones)-A$
		\STATE Compute second smallest eigenvector of $L_A$, $f_{1}$
		\STATE Sort the values of $f_{1}$
		\ENSURE Permutation $\sigma : f_{1}({\sigma(1)}) \leq 
		\ldots \leq f_{(1)}({\sigma(n)})$
	\end{algorithmic}
\end{algorithm}
\begin{theorem}[\cite{Atkins}]\label{thm:spectral-solves-seriation-preR}
If $A \in \symm_n$ is a pre-$\cR$ matrix, then Algorithm~\ref{alg:spectral} recovers a permutation $\Pi \in \cP_n$ such that $\PAP \in \cR^n$, \ie, it solves \ref{eqn:seriation}.
\end{theorem}

\subsection{Laplacian Embedding}

Let $0=\lambda_0 < \lambda_1 \leq \ldots \leq \lambda_{n-1}$,
$\Lambda \triangleq \diag\left(\lambda_0, \ldots, \lambda_{n-1} \right)$,
$\Phi = \left( \ones, f_{1}, \ldots, f_{n-1} \right)$,
be the eigendecomposition of $L_A = \Phi \Lambda \Phi^T$.
Algorithm~\ref{alg:spectral} embeds the data in 1D through the eigenvector $f_{1}$ (\kLE{1}).
For any $d < n$,
$\Phi^{(d)} \triangleq \left( f_{1}, \ldots, f_{d} \right)$
defines a $d$-dimensional embedding (\kLE{d})
\begin{align}\tag{\kLE{d}}\label{eqn:kLE}
    \yy_i = \left(f_{1}(i), f_{2}(i), \ldots, f_{d}(i)\right)^T \in \reals^d
    , \: \: \: \mbox{for} \: \: \: i=1,\ldots,n. 
\end{align}
which solves the following embedding problem,
\begin{align}
\tag{Lap-Emb} \label{eqn:lapl-embed}
\BA{ll}
\mbox{minimize}  & \sum_{i,j=1}^n 
A_{ij} \|\yy_i - \yy_j\|_{2}^2 \\
\st & \tilde{\Phi}=\left( \yy_{1}^T, \ldots, \yy_{n}^T \right)^T \in \reals^{n \times d} \:,\: \tilde{\Phi}^T \tilde{\Phi} = \idm_d \:,\: \tilde{\Phi}^T \ones_n = {\mathbf 0}_d
\EA
\end{align}
Indeed, like in \eqref{eqn:2sum-is-quadratic}, the objective of \ref{eqn:lapl-embed} can be written $\Tr \left( \tilde{\Phi}^T L_A \tilde{\Phi} \right)$ (see \citet{belkin2003laplacian} for a similar derivation).
The \ref{eqn:2sum} intuition still holds: the \kLE{d} lays similar elements nearby, and dissimilar apart, in $\reals^d$.
Other dimensionality reduction techniques such as Multidimensional scaling (MDS) \citep{kruskal1978multidimensional}, kernel PCA \citep{scholkopf1997kernel}, or Locally Linear Embedding (LLE) \citep{roweis2000nonlinear} could be used as alternatives to embed the data in a way that intuitively preserves the latent ordering. However, guided by the generalization of Algorithm~\ref{alg:spectral} and theoretical results that follow, we restrict ourselves to the Laplacian embedding.

\subsubsection{Normalization and Scaling}

Given the weighted adjacency matrix $W \in \symm_n$ of a graph, its Laplacian reads $L = D - W$, where $D = \diag(W\ones)$ has diagonal entries $d_{i} = \sum_{j=1}^n W_{ij}$ (degree of $i$).
Normalizing $W_{ij}$ by
$\sqrt{d_i d_j}$ or $d_i$
leads to the normalized Laplacians,
\begin{align}
\BA{rll}
\Lsym & = & D^{-1/2} L D^{-1/2} = \idm - D^{-1/2} W D^{-1/2}\\
\Lrw & = & D^{-1} L  =  \idm - D^{-1} W \\
\EA
\end{align}
They correspond to graph-cut normalization (normalized cut or ratio cut). Moreover,
$\Lrw$ has a Markov chain interpretation, where 
a random walker on edge $i$ jumps to edge $j$ from time $t$ to $t+1$ with transition probability $P_{ij} \triangleq W_{ij}/d_i$.
It has connections with diffusion processes, governed by the heat equation $\frac{\partial \cH_t}{\partial t} = - \Delta \cH_t$, where $\Delta$ is the Laplacian operator, $\cH_t$ the heat kernel, and $t$ is time \citep{qiu2007clustering}.
These connections lead to diverse Laplacian embeddings backed by theoretical justifications, where the eigenvectors $f^{\text{rw}}_{k}$ of $\Lrw$ are sometimes scaled by decaying weights $\alpha_k$ (thus emphasizing the first eigenvectors),
\begin{align}\tag{\kLE{($\alpha$, d)}}\label{eqn:scaled-kLE}
\tilde{\yy}_i = \left( \alpha_1 f_{1}^{\text{rw}}(i), \ldots, \alpha_{d-1} f_{d}^{\text{rw}}(i)\right)^T \in \reals^d
, \: \: \: \mbox{for} \: \: \: i=1,\ldots,n. 
\end{align}
Laplacian eigenmaps \citep{belkin2003laplacian} is a nonlinear dimensionality reduction technique based on the spectral embedding of $\Lrw$ (\eqref{eqn:scaled-kLE} with $\alpha_k =1$ for all $k$).
Specifically, given points $x_1, \ldots, x_n \in \reals^d$
, the method computes a heat kernel similarity matrix  $W_{ij} = \exp{-\left(\| x_i - x_j\|^2/t\right)}$ and outputs the first eigenvectors of $\Lrw$ as a lower dimensional embedding.
The choice of the heat kernel is motivated by connections with the heat diffusion process on a manifold, a partial differential equation involving the Laplacian operator.
This method has been successful in many machine learning applications such as semi-supervised classification \citep{belkin2004semi} and search-engine type ranking \citep{zhou2004ranking}. Notably, it provides a global, nonlinear embedding of the points that preserves the local structure.

The commute time distance $\text{CTD}(i,j)$ between two nodes $i$ and $j$ on the graph is the expected time for a random walker to travel from node $i$ to node $j$ and then return.
The full \ref{eqn:scaled-kLE}, with $\alpha_k = (\lambda_k^{\text{rw}})^{-1/2}$ and $d=n-1$, satisfies $\text{CTD}(i,j) \propto \| \tilde{\yy}_{i} - \tilde{\yy}_{j} \|$. Given the decay of $\alpha_k$, the \kLE{d} with $d \ll n$ approximately preserves the CTD.
This embedding has been successfully applied to vision tasks, \eg, anomaly detection \citep{albano2012euclidean}, image segmentation and motion tracking \citep{qiu2007clustering}.

Another, closely related dimensionality reduction technique is that of diffusion maps \citep{coifman2006diffusion}, where the embedding is derived to preserve diffusion distances, resulting in the \ref{eqn:scaled-kLE}, for $t \geq 0$, $\alpha_{k}(t) = (1 - \lambda_{k}^{\text{rw}})^t$.

\citet{coifman2006diffusion,coifman2008graph} also propose a normalization of the similarity matrix $\tilde{W} \gets D^{-1} W D^{-1}$, to extend the convergence of $\Lrw$ towards the Laplace-Beltrami operator on a curve when the similarity is obtained through a heat kernel on points that are \emph{non uniformly} sampled along that curve.

Finally, we will use in practice the heuristic scaling $\alpha_k = 1/\sqrt{k}$ to damp high dimensions, as explained in Appendix~\ref{ssec:scaling-sensitivity}.

For a deeper discussion about spectral graph theory and the relations between these methods, see for instance \citet{qiu2007clustering} and \citet{chung2000discrete}.

\subsection{Link with Continuous Operators}\label{ssec:asymptotic-lap}

In the context of dimensionality reduction, when the data points $x_1, \ldots, x_n \in \reals^D$ lie on a manifold $\mathcal{M} \subset \reals^d$ of dimension $K \ll D$, the graph Laplacian $L$ of the heat kernel ($W_{ij} = \exp{\left(-\| x_i - x_j\|^2/t\right)}$) used in \citet{belkin2003laplacian} is a discrete approximation of $\Delta_{\mathcal{M}}$, the Laplace-Beltrami operator on $\mathcal{M}$ (a differential operator akin to the Laplace operator, adapted to the local geometry of $\mathcal{M}$).
\citet{singer2006graph} specify the hypothesis on the data and the rate of convergence of $L$ towards $\Delta_{\mathcal{M}}$ when $n$ grows and the heat-kernel bandwidth $t$ shrinks.
\citet{von2005limits} also explore the spectral asymptotics of the spectrum of $L$ to prove consistency of spectral clustering.

This connection with continuous operators gives hints about the Laplacian embedding in some settings of interest for \ref{eqn:seriation} and \ref{eqn:circ-seriation}.
Indeed, consider $n$ points distributed along a curve $\Gamma \subset \reals^D$ of length $1$, parameterized by a smooth function $\gamma : \reals \rightarrow \reals^D$, 
$\Gamma = \{ \mathbf{\gamma}(s) \::\: s \in [0,1] \}$, say $x_i = \mathbf{\gamma}(i/n)$. If their similarity measures their proximity along the curve, then the similarity matrix is a circular-R matrix if the curve is closed ($\gamma(0)=\gamma(1)$), and a R matrix otherwise.
\citet{coifman2008graph} motivate a method for \ref{eqn:circ-seriation} with the spectrum of the Laplace-Beltrami operator $\Delta_{\Gamma}$ on $\Gamma$ when $\Gamma$ is a closed curve.
Indeed, $\Delta_{\Gamma}$ is simply the second order derivative with respect to the arc-length $s$, $\Delta_{\Gamma} f(s) = f^{\prime \prime}(s)$ (for $f$ twice continuously differentiable), and its eigenfunctions are given by,
\begin{align}\label{eqn:lapl-beltr-eigenpb}
    f^{\prime \prime}(s) = -\lambda f(s).
\end{align}
With periodic boundary conditions, $f(0)=f(1)$,  $f^{\prime}(0)=f^{\prime}(1)$, and smoothness assumptions, the first eigenfunction is constant with eigenvalue $\lambda_0=0$, and the remaining are $\left\{ \cos{\left(2 \pi m s \right)}, \: \sin{\left(2 \pi m s \right)} \right\}_{m=1}^{\infty}$, associated to the eigenvalues $\lambda_m = (2\pi m)^2$ of multiplicity 2.
Hence, the \kLE{2}, $\left( f_{1}(i), f_{2}(i) \right) \approx \left( \cos{(2\pi s_i)}, \sin{(2\pi s_i)} \right)$ should approximately lay the points on a circle, allowing for solving \ref{eqn:circ-seriation} \citep{coifman2008graph}.
More generally, the \kLE{2d}, $\left(f_{1}(i), \ldots, f_{2d+1}(i)\right)^T \approx \left( \cos{(2\pi s_i)}, \sin{(2\pi s_i)}, \ldots, \cos{(2 d \pi s_i)}, \sin{(2 d \pi s_i)} \right)$ is a closed curve in $\reals^{2d}$.

If $\Gamma$ is not closed, we can also find its eigenfunctions. For instance, with Neumann boundary conditions (vanishing normal derivative), say, $f(0)=1$, $f(1)=0$, $f^{\prime}(0)=f^{\prime}(1)=0$, the non-trivial eigenfunctions of $\Delta_{\Gamma}$ are $\left\{ \cos{\left(\pi m s \right)}\right\}_{m=1}^{\infty}$, with associated eigenvalues $\lambda_m = (\pi m)^2$ of multiplicity 1.
The \kLE{1} $f_{1}(i) \approx \cos{\left(\pi s_i \right)}$ respects the monotonicity of $i$, which is consistent with Theorem~\ref{thm:spectral-solves-seriation-preR}. \citet{lafon2004diffusion} invoked this asymptotic argument to solve an instance of \ref{eqn:seriation} but seemed unaware of the existence of Atkin's Algorithm~\ref{alg:spectral}.
Note that here too, the \kLE{d}, $\left(f_{1}(i), \ldots, f_{d}(i)\right)^T \approx \left( \cos{(\pi s_i)}, \ldots, \cos{(d \pi s_i)} \right)$ follows a closed curve in $\reals^{d}$, with endpoints.

These asymptotic results hint that the Laplacian embedding preserves the latent ordering of data points lying on a curve embedded in $\reals^D$.
However, these results are only asymptotic and there is no known guarantee for the \ref{eqn:circ-seriation} problem as there is for \ref{eqn:seriation}.
Also, the curve (sometimes called filamentary structure) stemming from the Laplacian embedding has been observed in more general cases where no hypothesis on a latent representation of the data is made, and the input similarity matrix is taken as is (see, \eg, \citet{diaconis2008horseshoes} for a discussion about the horseshoe phenomenon).

\subsection{Ordering points lying on a curve}
Finding the latent ordering of some points lying on (or close to) a curve can also be viewed as an instance of the traveling salesman problem (TSP), for which a plethora of (heuristic or approximation) algorithms exist \citep{reinelt1994traveling,laporte1992traveling}. We can think of this setting as one where the cities to be visited by the salesman are already placed along a single road, thus these TSP instances are easy and may be solved by simple heuristic algorithms.

Existing approaches for \ref{eqn:seriation} and \ref{eqn:circ-seriation} have only used 2D embeddings so far, for simplicity.
\citet{kuntz2001iterative} use the \kLE{2} to find a circular ordering of the data. They use a somehow exotic TSP heuristic which maps the 2D points onto a pre-defined ``space-filling'' curve, and unroll the curve through its closed form inverse to obtain a 1D embedding and sort the points.
\citet{friendly2002corrgrams} uses the angle between the two first coordinates of the 2D-MDS embedding and sorts them to perform \ref{eqn:seriation}.
\citet{coifman2008graph} use the \kLE{2} to perform \ref{eqn:circ-seriation} in a tomographic reconstruction setting,
and use a simple algorithm that sorts the inverse tangent of the angle between the two components to reorder the points.
\citet{liu2018unsupervised} use a similar approach to solve \ref{eqn:circ-seriation} in a cell-cycle related problem, but with the 2D embedding given by MDS.

%% file: sections/2_Laplacian_Embedding_circular.tex
\section{Spectral properties of some (circular) Robinson matrices}\label{sec:theory}

We have claimed that the \kLE{d} enhances the latent ordering of the data and we now present some theoretical evidences. We adopt a point of view similar to \citet{Atkins}, where the feasibility of \ref{eqn:seriation} relies on structural assumptions on the similarity matrix ($\cR$).
For a subclass $\SCR$ of $\circR$ (set of circular-R matrices), we show that the \kLE{d} lays the points on a closed curve, and that for $d=2$, the elements are embedded on a circle according to their latent circular ordering.
This is a counterpart of Theorem~\ref{thm:spectral-solves-seriation-preR} for \ref{eqn:circ-seriation}. It extends the asymptotic results motivating the approach of \citet{coifman2008graph}, shifting the structural assumptions on the elements (data points lying on a curve embedded in $\reals^D$) to assumptions on the raw similarity matrix that can be verified in practice.
Then, we develop a perturbation analysis to bound the deformation of the embedding when the input matrix is in $\SCR$ up to a perturbation.
Finally, we discuss the spectral properties of some (non circular) $\cR$-matrices that shed light on the filamentary structure of their \kLE{d} for $d>1$.

For simplicity, we assume $n \triangleq 2p+1$ odd in the following. The results with $n=2p$ even are relegated to the Appendix, together with technical proofs.
 
\subsection{Circular Seriation with Symmetric, Circulant matrices}

Let us consider the set $\SCR$ of matrices in $\circR$ that are circulant, in order to have a closed form expression of their spectrum.
A matrix $A \in \reals^{n \times n}$ is Toeplitz if its entries are constant on a given diagonal, $A_{ij} = b_{(i-j)}$ for a vector of values $b$ of size ${2n-1}$. A symmetric Toeplitz matrix $A$ satisfies $A_{ij} = b_{|i-j|}$, with $b$ of size ${n}$.
In the case of circulant symmetric matrices, we also have that $b_{k} = b_{n-k}$, for $1 \leq k \leq n$, thus symmetric circulant matrices are of the form,
\begin{align}
\label{eq:defrsptmat} A\ =\ \left(
\begin{array}{cccccc}
	b_0    &  b_1   &  b_2   &  \cdots  &  b_2  &  b_1  \\
	b_1    &  b_0   &  b_1   &  \cdots  &  b_3  &  b_2  \\
	b_2    &  b_1   &  b_0   &  \cdots  &  b_4  &  b_3  \\
	\vdots & \vdots & \vdots &  \ddots  & \vdots &  \vdots  \\
	b_2    & b_3    & b_4    &  \cdots  & b_0   &  b_1  \\
	b_1    & b_2    & b_3    &  \cdots  & b_1   &  b_0
\end{array}\right).
\end{align}
Where $b$ is a vector of values of size $p+1$ (recall that $n=2p+1$).
The circular-R assumption (Def~\ref{def:circ-R-mat}) imposes that the sequence $(b_0, \ldots, b_{p+1})$ is non-increasing.
We thus define the set $\SCR$ of circulant matrices of $\circR$ as follows.
\begin{definition}
A matrix $A \in \symm^n$ is in $\SCR$ iff it verifies $A_{ij}=b_{|i-j|}$ and $b_k = b_{n-k}$ for $1\leq k\leq n$ with $(b_k)_{k=0,\ldots,\lfloor n/2 \rfloor}$ a non-increasing sequence.
\end{definition}

The spectrum of symmetric circulant  matrices is known \citep{reichel1992eigenvalues,gray2006toeplitz,massey2007distribution}, and for a matrix $A$ of size $n=2p+1$, it is given by,
\begin{align}\label{eqn:spectrum-circ}
\BA{lll}
\vspace{.1cm}
    \nu_m  &=&  b_0 + 2 {\textstyle \sum_{k=1}^{p}} {b_k \cos{ \left(2 \pi k m/n\right)}}\\
\vspace{.1cm}
y^{m, \cos} &= & \frac{1}{\sqrt{n}} \left(1, \cos \left( 2 \pi m / n \right), \ldots, \cos \left( 2 \pi m (n-1) / n \right) \right)\\
\vspace{.1cm}
y^{m, \sin} &= & \frac{1}{\sqrt{n}} \left(1, \sin \left( 2 \pi m / n \right), \ldots, \sin \left( 2 \pi m (n-1) / n \right) \right)~.
\EA
\end{align}
For $m = 1,\ldots,p$, $\nu_m$ is an eigenvalue of multiplicity 2 with associated eigenvectors $y^{m, \cos}$,$y^{m, \sin}$.
For any $m$, $(y^{m, \cos}, y^{m, \sin})$ embeds the points on a circle, but for $m>1$, the circle is walked through $m$ times, hence the ordering of the points on the circle does not follow their latent ordering.
The $\nu_m$ from equations~\eqref{eqn:spectrum-circ} are in general not sorted. It is the Robinson property (monotonicity of $(b_k)$) that guarantees that $\nu_1 \geq \nu_m$, for $m \geq 1$, and thus that
the \kLE{2} embeds the points on a circle \emph{that follows the latent ordering} and allows one to recover it by scanning through the unit circle.
This is formalized in Theorem~\ref{th:without_noise}, which is the main result of our paper, proved in Appendix \ref{sec:circular_toeplitz_matrix}.
It provides guarantees in the same form as in Theorem~\ref{thm:spectral-solves-seriation-preR} with the simple Algorithm~\ref{alg:circular-2d-ordering} that sorts the angles, used in \citet{coifman2008graph}.
\begin{algorithm}[H]
	\caption{Circular Spectral Ordering \citep{coifman2008graph}}\label{alg:circular-2d-ordering}
	\begin{algorithmic} [1]
		\REQUIRE Connected similarity matrix $A \in \mathbb{R}^{n \times n}$
		\STATE Compute normalized Laplacian $\Lrw_{A}= \idm - \left(\diag(A\ones)\right)^{-1}A$
		\STATE Compute the two first non-trivial eigenvectors of $\Lrw_A$, $\left(f_{1}, f_{2}\right)$
		\STATE Sort the values of $\theta(i) \triangleq   \tan^{-1}{\left(f_{2}(i)/f_{1}(i) \right)} + \mathbbm{1}[f_1(i)<0] \pi$
		\ENSURE Permutation $\sigma : \theta({\sigma(1)}) \leq 
		\ldots \leq \theta({\sigma(n)})$
	\end{algorithmic}
\end{algorithm}
\begin{theorem}\label{th:without_noise}
Given a permuted observation $\PAP$ ($\Pi \in \cP$) of a matrix $A \in \SCR$, the \kLE{2} maps the items on a circle, equally spaced by angle $2\pi/n$, following the circular ordering in $\Pi$.
Hence, Algorithm~\ref{alg:circular-2d-ordering} recovers a permutation $\Pi \in \cP_n$ such that $\PAP \in \SCR$, \ie, it solves \ref{eqn:circ-seriation}.
\end{theorem}

\subsection{Perturbation analysis}

The spectrum is a continuous function of the matrix.
Let us bound the deformation of the \kLE{2} under a perturbation of the matrix $A$ using
the Davis-Kahan theorem \citep{davis1970rotation}, well introduced in \citep[Theorem 7]{von2007tutorial}.
We give more detailed results in Appendix~\ref{sec:perturbation_analysis} for a subclass of $\SCR$ (KMS) defined further.
\begin{proposition}[Davis-Kahan]\label{prop:davis_Kahan}
Let $L$ and $\tilde{L} = L + \delta L$ be the Laplacian matrices of $A \in \SCR$ and $A + \delta A \in \symm^n$, respectively, and $V,\tilde{V} \in \reals^{2 \times n}$ be the associated \kLE{2} of $L$ and $\tilde{L}$, \ie, the concatenation of the two eigenvectors associated to the two smallest non-zero eigenvalues, written $\lambda_1 \leq \lambda_2$ for $L$.
Then, there exists an orthonormal rotation matrix $O$ such that
\begin{eqnarray}\label{eq:perturbation_result}
\frac{\|V_1-\tilde{V}_1 O\|_F}{\sqrt{n}} \leq  \frac{\|\delta A\|_F}{\min(\lambda_1,\lambda_2-\lambda_1)}~.
\end{eqnarray}
\end{proposition}

\subsection{Robinson Toeplitz matrices}

Let us investigate how the latent linear ordering of Toeplitz matrices in $\cR$ translates to the \kLE{d}.
Remark that from Theorem~\ref{thm:spectral-solves-seriation-preR}, the \kLE{1} suffices to solve \ref{eqn:seriation}.
Yet, for perturbed observations of $A \in \cR$, the \kLE{d} may be more robust to the perturbation than the \kLE{1}, as the experiments in~\S\ref{sec:numerical_results} indicate.

\textbf{Tridiagonal Toeplitz matrices}
are defined by $b_0 > b_1 > 0=b_2 = \ldots = b_p $.
For $m=0,\ldots,n-1$, they have eigenvalues $\nu_m$ with multiplicity 1 associated to eigenvector $y^{(m)}$ \citep{trench1985eigenvalue},
\begin{align}\label{eqn:spectrum-tridiag}
	\BA{lll}
    \nu_m & =& b_0 + 2 b_1 \cos{\left(m \pi / (n+1)\right)}\\
     y^{(m)} &=& \left( \sin{\left(m \pi / (n+1)\right)}, \ldots,  \sin{\left(m n \pi / (n+1)\right)} \right),~
     \EA
\end{align}
thus matching the spectrum of the Laplace operator on a curve with endpoints from \S\ref{ssec:asymptotic-lap} (up to a shift).
This type of matrices can indeed be viewed as a limit case with points uniformly sampled on a line with strong similarity decay, leaving only the two nearest neighbors with non-zero similarity. 

\textbf{Kac-Murdock-Szegö (KMS) matrices}
are defined, for $\alpha > 0$, $\rho = e^{-\alpha}$, by $A_{ij} = b_{|i-j|} = e^{-\alpha |i-j|} = \rho^{|i-j|}$.
For $m=1,\ldots,\lfloor n/2 \rfloor$, there exists $\theta_m \in \left({(m-1)\pi}/{n}, {m\pi}/{n}\right)$, such that $\nu_m$ is  a double eigenvalue associated to eigenvectors $y^{m, \cos}$,$y^{m, \sin}$,
\begin{align}\label{eqn:spectrum-KMS}
\BA{lll}
\nu_m  &=&  \frac{1 - \rho^2}{1 - 2 \rho \cos{\theta_m} + \rho^2}\\
\vspace{.1cm}
y^{m, \cos} &= & \left( \cos{\left( (n-2 r +1)\theta_m /2 \right)} \right)_{r=1}^{n}\\
\vspace{.1cm}
y^{m, \sin} &= & \left( \sin{\left( (n-2 r +1)\theta_m /2 \right)} \right)_{r=1}^{n}~.
\EA
\end{align}

\textbf{Linearly decreasing Toeplitz matrices}
defined by $A^{lin}_{ij} = b_{|i-j|} = n - |i-j|$ have spectral properties analog to those of KMS matrices (trigonometric expression, interlacement, low frequency assigned to largest eigenvalue), but with more technical details available in \citet{bunger2014inverses}.
This goes beyond the asymptotic case modeled by tridiagonal matrices.

\textbf{Banded Robinson Toeplitz matrices} typically include
similarity matrices from DNA sequencing.
Actually, any Robinson Toeplitz matrix becomes banded under a thresholding operation.
Also, fast decaying Robinson matrices such as KMS matrices are almost banded.
There is a rich literature dedicated to the spectrum of generic banded Toeplitz matrices \citep{boeottcher2005spectral,gray2006toeplitz,bottcher2017asymptotics}.
However, it mostly provides asymptotic results on the spectra.
Notably, some results indicate that the eigenvectors of some banded symmetric Toeplitz matrices become, up to a rotation, close to the sinusoidal, almost equi-spaced eigenvectors observed in equations~\eqref{eqn:spectrum-tridiag} and \eqref{eqn:spectrum-KMS} \citep{bottcher2010structure,ekstrom2017eigenvalues}.

\subsection{Spectral properties of the Laplacian}\label{ssec:spectral-prop-lapl}
For circulant matrices $A$, $L_A$ and $A$ have the same eigenvectors since $L_A = \diag (A \ones) - A = c \idm - A$, with $c \triangleq \sum_{k=0}^{n-1} b_k$. For general symmetric Toeplitz matrices, this property no longer holds as $c_i = \sum_{j=1}^{n} b_{|i-j|}$ varies with $i$.
Yet, for fast decaying Toeplitz matrices, $c_i$ is almost constant except for $i$ at the edges, namely $i$ close to $1$ or to $n$.
Therefore, the eigenvectors of $L_A$ resemble those of $A$ except for the ``edgy'' entries.

%% file: sections/3_Recovery_algorithm.tex
\section{Recovering Ordering on Filamentary Structure}\label{sec:results}

We have seen that (some) similarity matrices $A$ with a latent ordering lead to a filamentary \kLE{d}.
The \kLE{d} integrates local proximity constraints together into a global consistent embedding. We expect isolated (or, uncorrelated) noise on $A$ to be averaged out by the spectral picture.
Therefore, we present Algorithm~\ref{algo:Recovery_order_filamentary} that redefines the similarity $S_{ij}$ between two items from their proximity within the \kLE{d}. Basically, it fits the points by a line \emph{locally}, in the same spirit as LLE, which makes sense when the data lies on a linear manifold (curve) embedded in $\reals^K$.
Note that Spectral Ordering (Algorithm~\ref{alg:spectral}) projects all points on a given line (it only looks at the first coordinates $f_{1}(i)$) to reorder them. Our method does so in a local neighborhood, allowing for reordering points on a curve with several oscillations.
We then run the basic Algorithms~\ref{alg:spectral} (or~\ref{alg:circular-2d-ordering} for \ref{eqn:circ-seriation}).
Hence, the \kLE{d} is eventually used to pre-process the similarity matrix.
\setlength{\textfloatsep}{10pt}
\begin{algorithm}[ht]
	\footnotesize
	\caption{Ordering Recovery on Filamentary Structure in $\reals^K$.}\label{algo:Recovery_order_filamentary}
	\begin{algorithmic} [1]
        \REQUIRE  A similarity matrix $A\in\mathcal{S}_n$, a neighborhood size $k \geq 2$, a dimension of the Laplacian Embedding $d$.
        \STATE $\Phi =\left( \yy_{1}^T, \ldots, \yy_{n}^T \right)^T \in \reals^{n \times d}\gets \text{\kLE{d}}(A)$ \hfill $\triangleright$ \text{Compute Laplacian Embedding}\label{line:laplacian_embedding_enhance_filament}
		\STATE Initialize $S = \idm_n$ \hfill $\triangleright$ \text{New similarity matrix}
		\FOR{$i=1,\ldots,n$}
			\STATE $V \gets  \{j \: : \: j \in k\text{-NN}(\yy_i) \}\cup \{i\}$ \hfill $\triangleright$ \text{find $k$ nearest neighbors of $\yy_i \in \reals^d$} \label{line:find_neighborhood}
			\STATE $w \gets  \text{LinearFit}(V)$ \hfill $\triangleright$ \text{fit $V$ by a line }
			\label{line:find_direction}
			\STATE $D_{uv} \gets |w^T (\yy_u - \yy_v) |$, for $u,v \in V$. \hfill \text{$\triangleright$ Compute distances on the line}
			\STATE  $S_{uv} \gets S_{uv} + D_{uv}^{-1}$, for $u,v \in V$. \hfill \text{$\triangleright$ Update similarity} \label{line:update-similarity}
		\ENDFOR
		\STATE Compute $\sigma^*$ from the matrix $S$ with Algorithm~\ref{alg:spectral} (resp., Algorithm~\ref{alg:circular-2d-ordering}) for a linear (resp., circular) ordering. \label{line:laplacian_embedding_enhance_basic_ordering}
		\ENSURE A permutation $\sigma^*$. 
	\end{algorithmic}
\end{algorithm}

In Algorithm~\ref{algo:Recovery_order_filamentary}, we compute a \kLE{d} in line~\ref{line:laplacian_embedding_enhance_filament} and then a \kLE{1} (resp., a \kLE{2}) for linear ordering (resp., a circular ordering) in line~\ref{line:laplacian_embedding_enhance_basic_ordering}. For reasonable number of neighbors $k$  in the $k$-NN of line 4 (in practice, $k=10$), the complexity of computing the \kLE{d} dominates Algorithm~\ref{algo:Recovery_order_filamentary}. We shall see in Section~\ref{sec:numerical_results} that our method, while being almost as computationally cheap as the base Algorithms~\ref{alg:spectral} and \ref{alg:circular-2d-ordering} (roughly only a factor 2), yields substantial improvements.
In line~\ref{line:update-similarity} we can update the similarity $S_{uv}$ by adding any non-increasing function of the distance $D_{uv}$, \eg, $D_{uv}^{-1}$, $\exp{\left(-D_{uv}\right)}$, or $-D_{uv}$ (the latter case requires to add an offset to $S$ afterwards to ensure it has non-negative entries. It is what we implemented in practice.)
In line~\ref{line:laplacian_embedding_enhance_basic_ordering}, the matrix $S$ needs to be connected in order to use Algorithm~\ref{alg:spectral}, which is not always verified in practice (for low values of $k$, for instance).
In that case, we reorder separately each connected component of $S$ with Algorithm~\ref{alg:spectral}, and then merge the partial orderings into a global ordering by using the input matrix $A$, as detailed in Algorithm~\ref{alg:connect_clusters}, Appendix~\ref{sec:appendix_algo}.

%% file: sections/4_numerical_results.tex
\section{Numerical Results}\label{sec:numerical_results}
\subsection{Synthetic Experiments}\label{ssec:synthetic-exps}
We performed synthetic experiments with noisy observations of Toeplitz matrices $A$, either linear ($\cR$) or circular ($\SCR$). We added a uniform noise on all the entries, with an amplitude parameter $a$ varying between 0 and $5$, with maximum value of the noise $a \| A \|_F$. The matrices $A$ used are either banded (sparse), with linearly decreasing entries when moving away from the diagonal, or dense, with exponentially decreasing entries (KMS matrices).
We used $n=500$, several values for the parameters $k$ (number of neighbors) and $d$ (dimension of the \kLE{d}), and various scalings of the \kLE{d} (parameter $\alpha$ in~\ref{eqn:scaled-kLE}), yielding similar results (see sensitivity to the number of neighbors $k$ and to the scaling~\ref{eqn:scaled-kLE} in Appendix~\ref{ssec:app-k-sensitivity}).
In an given experiment, the matrix $A$ is randomly permuted with a ground truth permutation $\pi^*$. We report the Kendall-Tau scores between $\pi^*$ and the solution of Algorithm~\ref{algo:Recovery_order_filamentary} for different choices of dimension $K$, for varying noise amplitude $a$, in Figure~\ref{fig:exp-main-banded}, for banded (circular) matrices.
For the circular case, the ordering is defined up to a shift. To compute a Kendall-Tau score from two permutations describing a circular ordering, we computed the best Kendall-Tau scores between the first permutation and all shifts from the second, as detailed in Algorithm~\ref{alg:circular-kendall-tau}.
The analog results for exponentially decaying (KMS) matrices are given in Appendix~\ref{ssec:app-KMS-main-exp}, Figure~\ref{fig:exp-main-KMS}.
For a given combination of parameters, the scores are averaged on 100 experiments and the standard-deviation divided by $\sqrt{n_{\text{exps}}} = 10$ (for ease of reading) is plotted in transparent above and below the curve.
The baseline (in blue) corresponds to the basic spectral method of Algorithm~\ref{alg:spectral} for linear and Algorithm~\ref{alg:circular-2d-ordering} for circular seriation.
Other lines correspond to given choices of the dimension of the \kLE{d}, as written in the legend.
\begin{figure}[hbt]
	\begin{center}
		\begin{subfigure}[htb]{0.45\textwidth}
			\includegraphics[width=\textwidth]{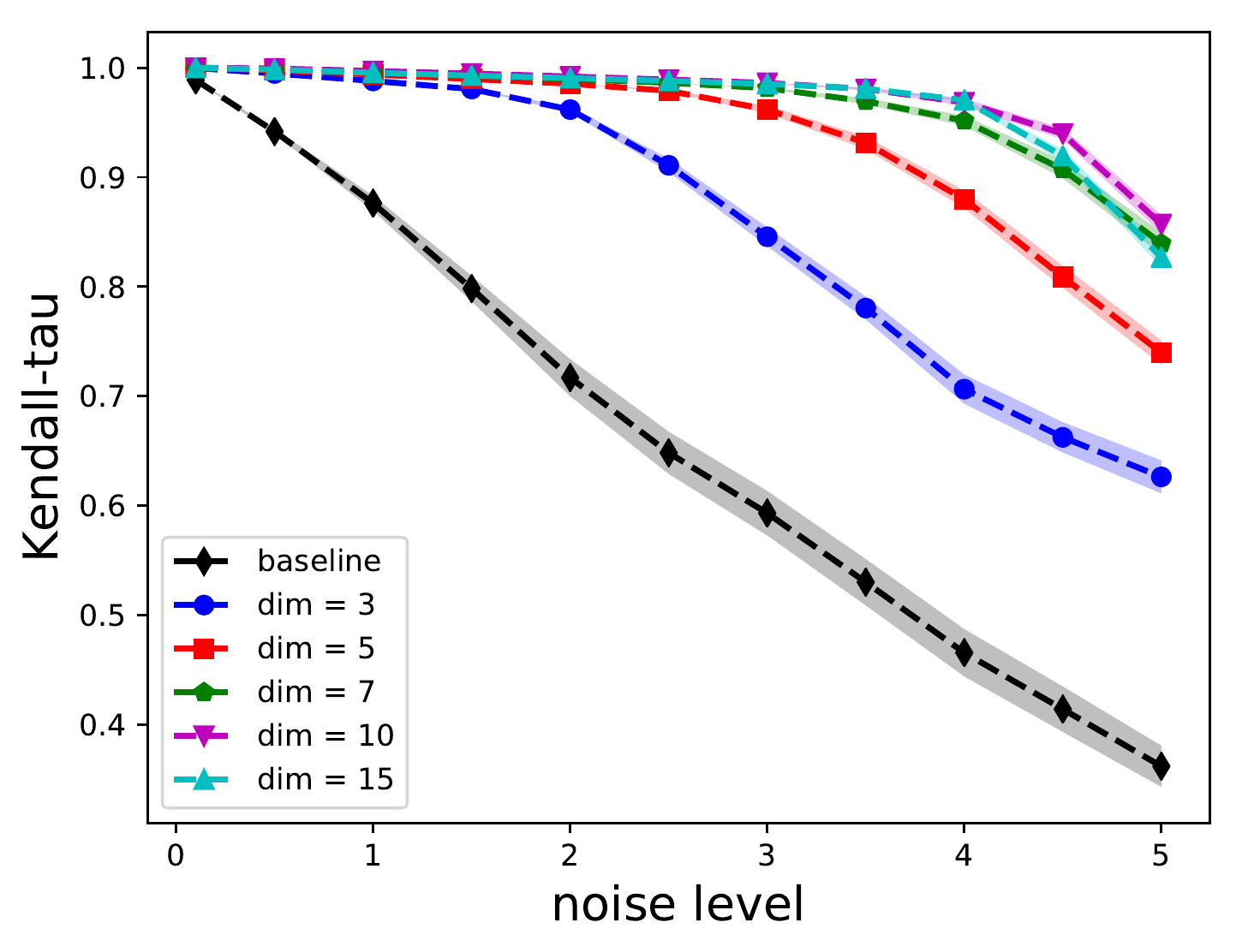}
			\caption{Linear Banded}\label{subfig:exps-lin-banded}
		\end{subfigure}
		\begin{subfigure}[htb]{0.45\textwidth}
			\includegraphics[width=\textwidth]{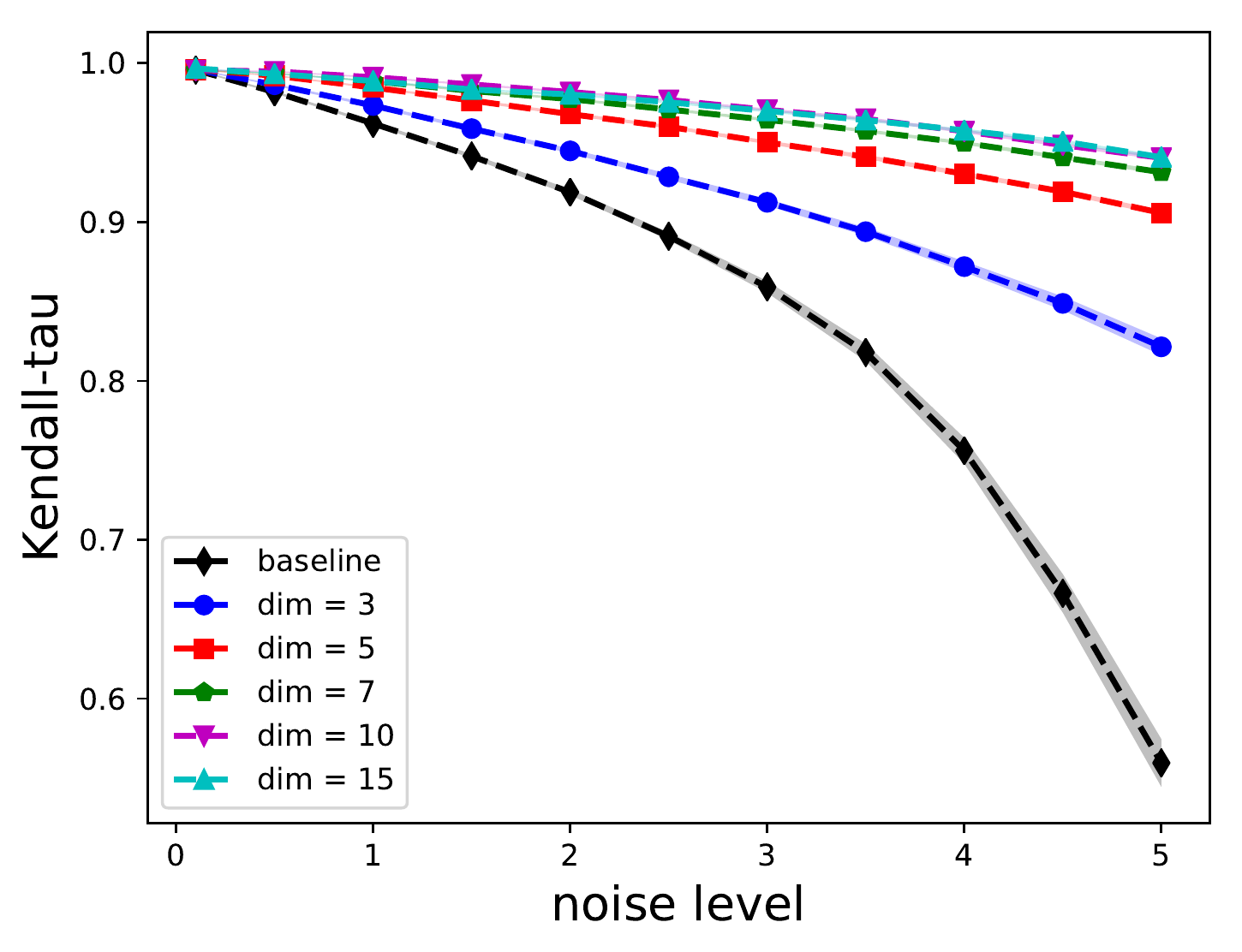}
			\caption{Circular Banded}\label{subfig:exps-circ-banded}
		\end{subfigure}
		\caption{
		Kendall-Tau scores for Linear (\ref{subfig:exps-lin-banded}) and Circular (\ref{subfig:exps-circ-banded}) Seriation for noisy observations of banded, Toeplitz, matrices, displayed for several values of the dimension parameter of the \kLE{d}($d$), for fixed number of neighbors $k=15$.
		}
		\label{fig:exp-main-banded}
	\end{center}
	\vskip -0.2in
\end{figure}

We observe that leveraging the additional dimensions of the \kLE{d} unused by the baseline methods Algorithm~\ref{alg:spectral} and~\ref{alg:circular-2d-ordering} substantially improves the robustness of Seriation. For instance, in Figure~\ref{subfig:exps-lin-banded}, the performance of Algorithm~\ref{algo:Recovery_order_filamentary} is almost optimal for a noise amplitude going from 0 to 4, when it falls by a half for Algorithm~\ref{alg:spectral}.
We illustrate the effect of the pre-processing of Algorithm~\ref{algo:Recovery_order_filamentary} in Figures~\ref{fig:exp-noisy-illustration} and~\ref{fig:exp-clean-illustration}, Appendix~\ref{ssec:illustrations}.

\subsection{Genome assembly experiment}
In {\it de novo} genome assembly, a whole DNA strand is reconstructed from randomly sampled sub-fragments (called {\it reads}) whose positions within the genome are unknown. The genome is oversampled so that all parts are covered by multiple reads with high probability. 
Overlap-Layout-Consensus (OLC) is a major assembly paradigm based on three main steps.
First, compute the overlaps between all pairs of read. This provides a similarity matrix $A$, whose entry $(i,j)$ measures how much reads $i$ and $j$ overlap (and is zero if they do not).
Then, determine the layout from the overlap information, that is to say find an ordering and positioning of the reads that is consistent with the overlap constraints. This step, akin to solving a one dimensional jigsaw puzzle, is a key step in the assembly process.
Finally, given the tiling of the reads obtained in the layout stage, the consensus step aims at determining the most likely DNA sequence that can be explained by this tiling. It essentially consists in performing multi-sequence alignments.

In the true ordering (corresponding to the sorted reads' positions along the genome), a given read overlaps much with the next one, slightly less with the one after it, and so on, until a point where it has no overlap with the reads that are further away. This makes the read similarity matrix Robinson and roughly band-diagonal (with non-zero values confined to a diagonal band). 
Finding the layout of the reads therefore fits the \ref{eqn:seriation} framework (or \ref{eqn:circ-seriation} for circular genomes, as illustrated in Supplementary Figure~\ref{fig:circular-genome-illustration}).
In practice however, there are some repeated sequences (called {\it repeats}) along the genome that induce false positives in the overlap detection tool \citep{Pop04}, resulting in non-zero similarity values outside (and possibly far away) from the diagonal band. The similarity matrix ordered with the ground truth is then the sum of a Robinson band matrix and a sparse ``noise'' matrix, as in Figure~\ref{subfig:ecoli-sim-mat}.
Because of this sparse ``noise'', the basic spectral Algorithm~\ref{alg:spectral} fails to find the layout, as the quadratic loss appearing in \ref{eqn:2sum} is sensitive to outliers.
\citet{recanati2018robust} tackle this issue by modifying the loss in \ref{eqn:2sum} to make it more robust.
Instead, we show that the simple multi-dimensional extension proposed in Algorithm~\ref{algo:Recovery_order_filamentary} suffices to capture the ordering of the reads despite the repeats.

\begin{figure}[hbt]
	\begin{center}
		\begin{subfigure}[htb]{0.45\textwidth}
			\includegraphics[width=\textwidth]{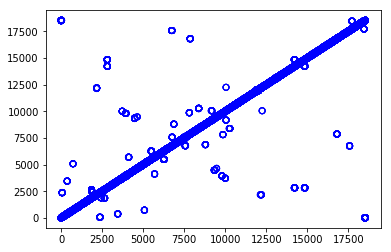}
			\caption{similarity matrix}\label{subfig:ecoli-sim-mat}
		\end{subfigure}
		\begin{subfigure}[htb]{0.45\textwidth}
			\includegraphics[width=\textwidth]{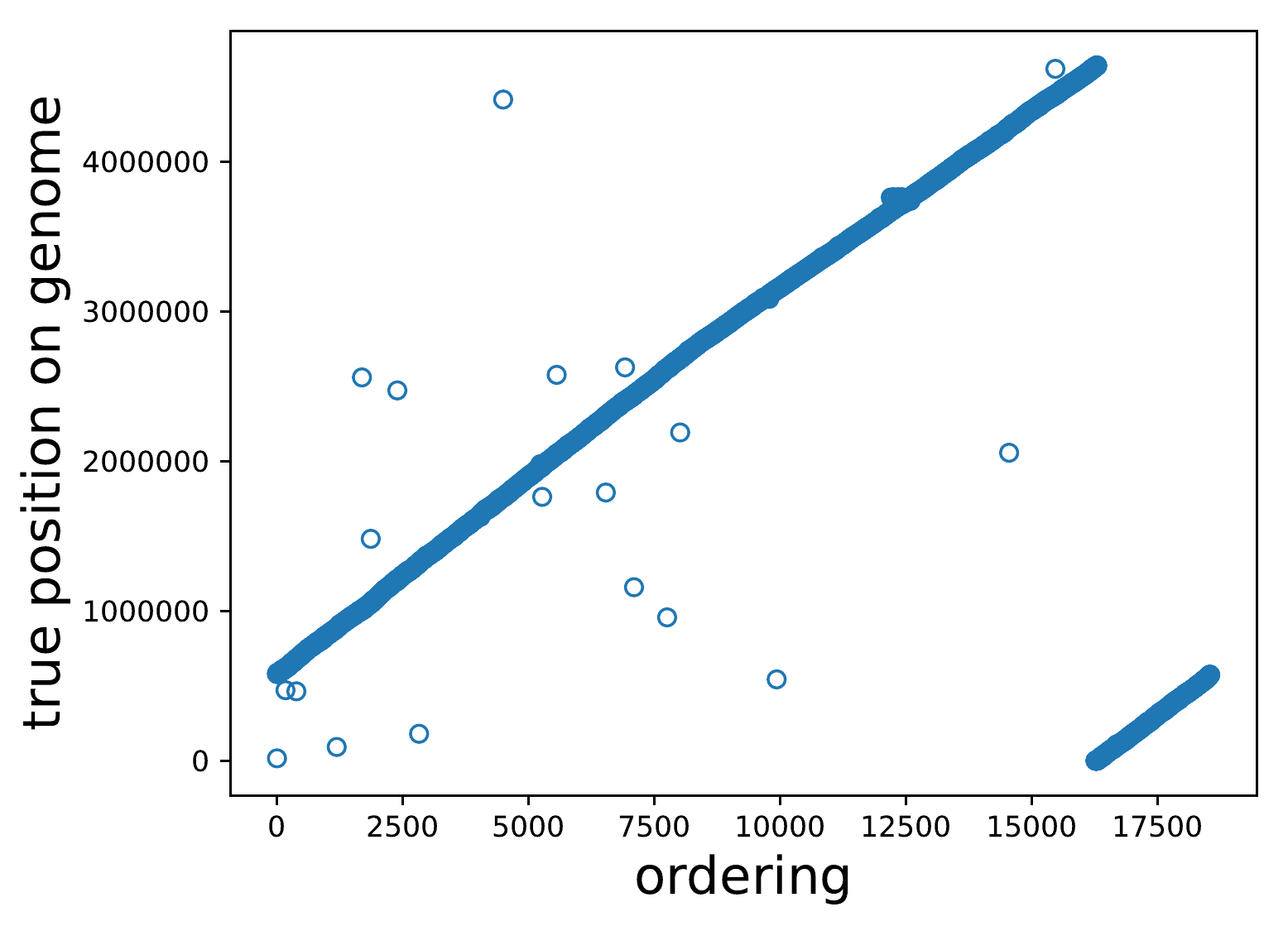}
			\caption{ordering found}\label{subfig:ecoli-total-ordering}
		\end{subfigure}
		\caption{
			Overlap-based similarity matrix (\ref{subfig:ecoli-sim-mat}) from {\it E. coli} reads, and the ordering found with Algorithm~\ref{algo:Recovery_order_filamentary} (\ref{subfig:ecoli-total-ordering}) versus the position of the reads within a reference genome obtained by mapping to a reference with minimap2.
			The genome being circular, the ordering is defined up to a shift, which is why we observe two lines instead of one in (\ref{subfig:ecoli-total-ordering}).
		}
		\label{fig:ecoli-exp}
	\end{center}
	\vskip -0.2in
\end{figure}

We used our method to perform the layout of a {\it E. coli} bacterial genome. We used reads sequenced with third-generation sequencing data, and computed the overlaps with dedicated software, as detailed in Appendix~\ref{ssec:supp-genome-assembly-exp}.
The new similarity matrix $S$ computed from the embedding in Algorithm~\ref{algo:Recovery_order_filamentary} was disconnected, resulting in several connected component instead of one global ordering (see Figure~\ref{subfig:ecoli-partial-orderings}). However, the sub-orderings could be unambiguously merged into one in a simple way described in Algorithm~\ref{alg:connect_clusters}, resulting in the ordering shown in Figure~\ref{subfig:ecoli-total-ordering}.
The Kendall-Tau score between the ordering found and the one obtained by sorting the position of the reads along the genome (obtained by mapping the reads to a reference with minimap2 \citep{li2018minimap2}) is of 99.5\%, using Algorithm~\ref{alg:circular-kendall-tau} to account for the circularity of the genome.

\section{Conclusion}
In this paper, we bring together results that shed light on the filamentary structure of the Laplacian embedding of serial data.
It allows for tackling~\ref{eqn:seriation} and~\ref{eqn:circ-seriation} in a unifying framework.
Notably, we provide theoretical guarantees for~\ref{eqn:circ-seriation} analog to those existing for~\ref{eqn:seriation}.
These do not make assumptions about the underlying generation of the data matrix, and can be verified {\it a posteriori} by the practitioner.
Then, we propose a simple method to leverage the filamentary structure of the embedding.
It can be seen as a pre-processing of the similarity matrix.
Although the complexity is comparable to the baseline methods, experiments on synthetic and real data indicate that this pre-processing substantially improves robustness to noise.

%% file: sections/8_appendix_algorithms.tex
\textbf{Notation:} We will commonly denote $\sigma$ a permutation of $\{1,\ldots,n\}$ and $\mathfrak{S}$ the set of all such permutations. When represented matricially, $\sigma$ will often be noted $\Pi$ while cyclic permutation of $\{1,\ldots,n\}$ will be noted as $\tau$. $A$ will usually denote the matrix of raw pair-wise similarities. $S$ will denote the similarity matrix resulting from Algorithm~\ref{algo:Recovery_order_filamentary}, and $k$ a neighboring parameter. Finally we use indexed version $\nu$ (resp., $\lambda$) to denote eigenvalues of a similarity matrix (resp. a graph Laplacian).

\section{Additional Algorithms}\label{sec:appendix_algo}

\subsection{Merging connected components}
The new similarity matrix $S$ computed in Algorithm~\ref{algo:Recovery_order_filamentary} is not necessarily the adjacency matrix of a connected graph, even when the input matrix $A$ is. For instance, when the number of nearest neighbors $k$ is low and the points in the embedding  are non uniformly sampled along a curve, $S$ may have several, disjoint connected components (let us say there are $C$ of them in the following).
Still, the baseline Algorithm~\ref{alg:spectral} requires a connected similarity matrix as input.
When $S$ is disconnected, we run \ref{alg:spectral} separately in each of the $C$ components, yielding $C$ sub-orderings instead of a global ordering.

However, since $A$ is connected, we can use the edges of $A$ between the connected components to merge the sub-orderings together.
Specifically, given the $C$ ordered subsequences, we build a meta similarity matrix between them as follows. For each pair of ordered subsequences $(c_i, c_j)$, we check whether the elements in one of the two ends of $c_i$ have edges with those in one of the two ends of $c_j$ in the graph defined by $A$. According to that measure of similarity and to the direction of these meta-edges (\ie, whether it is the beginning or the end of $c_i$ and $c_j$ that are similar), we merge together the two subsequences that are the closest to each other. We repeat this operation with the rest of the subsequences and the sequence formed by the latter merge step, until there is only one final sequence, or until the meta similarity between subsequences is zero everywhere.
We formalize this procedure in the greedy Algorithm~\ref{alg:connect_clusters}, which is implemented  in the package at \url{https://github.com/antrec/mdso}.

Given $C$ reordered subsequences (one per connected component of $S$) $(c_i)_{i=1,\ldots,C}$, that form  a partition of $\{1,\ldots,n\}$, and a window size $h$ that define the length of the ends we consider ($h$ must be smaller than half the smallest subsequence), we denote by $c_i^-$ (resp. $c_i^+$) the first (resp. the last) $h$ elements of $c_i$, and $a(c_i^{\epsilon}, c_j^{\epsilon^\prime}) = \sum_{u \in c_i^{\epsilon}, v \in c_j^{\epsilon^\prime}} A_{uv}$ is the similarity between the ends $c_i^{\epsilon}$ and $c_j^{\epsilon^\prime}$, for any pair $c_i, c_j$,  $i \neq j \in \{1, \ldots, C\}$, and any combination of ends $\epsilon, \epsilon^\prime \in \{+,-\}$.
Also, we define the meta-similarity between $c_i$ and $c_j$ by,
\begin{eqnarray}\label{eqn:meta-sim-merge}
s(c_i,c_j) \triangleq \text{max}(a(c_i^+,c_j^+), a(c_i^+,c_j^-), a(c_i^-,c_j^+), a(c_i^-,c_j^-))~,
\end{eqnarray}
and $(\epsilon_i, \epsilon_j) \in \{+,-\}^2$ the combination of signs where the argmax is realized, \ie, such that $s(c_i,c_j) = a(c_i^{\epsilon_i}, c_j^{\epsilon_j})$.
Finally, we will use $\bar{c}_i$ to denote the ordered subsequence $c_i$ read from the end to the beginning, for instance if $c=(1,\ldots, n)$,  then $\bar{c} = (n, \ldots, 1)$.

\begin{algorithm}[H]
	\caption{Merging connected components}\label{alg:connect_clusters}
	\begin{algorithmic} [1]
		\REQUIRE $C$ ordered subsequences forming a partition $P=(c_1,\ldots,c_C)$ of $\{1,\ldots,n\}$, an initial similarity matrix $A$, a neighborhood parameter $h$. 
		\WHILE{$C > 1$}
			\STATE Compute meta-similarity $\tilde{S}$ such that $\tilde{S}_{ij} = s(c_i, c_j)$, and meta-orientation $(\epsilon_i, \epsilon_j)$, for all pairs of subsequences with equation~\ref{eqn:meta-sim-merge}.
			\IF{$\tilde{S} = 0$} \STATE break \ENDIF \label{line:break-merge}
			\STATE find $(i,j) \in \argmax \tilde{S}$, and $(\epsilon_i, \epsilon_j)$ the corresponding orientations.
			\IF{$(\epsilon_i, \epsilon_j) = (+,-)$}
				\STATE $c^{\text{new}} \gets (c_i, c_j)$
			\ELSIF{$(\epsilon_i, \epsilon_j) = (+,+)$}
				\STATE $c^{\text{new}} \gets (c_i, \bar{c}_j)$
			\ELSIF{$(\epsilon_i, \epsilon_j) = (-,-))$}
				\STATE $c^{\text{new}} \gets (\bar{c}_i, c_j)$
			\ELSIF{$(\epsilon_i, \epsilon_j) = (-,+))$}
				\STATE $c^{\text{new}} \gets (\bar{c}_i, \bar{c}_j)$
			\ENDIF
			\STATE Remove $c_i$ and $c_j$ from $P$.
			\STATE Add $c^{\text{new}}$ to $P$.
			\STATE $C \gets C - 1$
		\ENDWHILE
		\ENSURE Total reordered sequence $c^{\text{final}}$, which is a permutation if $C=1$ or a set of reordered subsequences if the loop broke at line~\ref{line:break-merge}.
	\end{algorithmic}
\end{algorithm}

\subsection{Computing Kendall-Tau score between two permutations describing a circular ordering}
Suppose we have data having a circular structure, \ie, we have $n$ items that can be laid on a circle such that the higher the similarity between two elements is, the closer they are on the circle.
Then, given an ordering of the points that respects this circular structure (\ie, a solution to \ref{eqn:circ-seriation}), we can shift this ordering without affecting the circular structure.
For instance, in Figure~\ref{fig:circular-shift-illustration}, the graph has a $\circR$ affinity matrix whether we use the indexing printed in black (outside the circle), or a shifted version printed in purple (inside the circle).
\begin{figure}[hbt]
	\begin{center}
		\includegraphics[width=.35\textwidth]{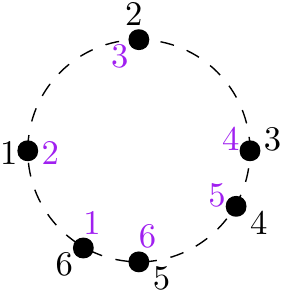}
		\caption{
			Illustration of the shift-invariance of permutations solution to a \ref{eqn:circ-seriation} problem.
		}
		\label{fig:circular-shift-illustration}
	\end{center}
\end{figure}
Therefore, we transpose the Kendall-Tau score between two permutations to the case where we want to compare the two permutations up to a shift with Algorithm~\ref{alg:circular-kendall-tau}
\begin{algorithm}[H]
	\caption{Comparing two permutation defining a circular ordering}\label{alg:circular-kendall-tau}
	\begin{algorithmic} [1]
		\REQUIRE Two permutations vectors of size $n$, $\sigma = \left(\sigma(1), \ldots, \sigma(n)\right)$ and $\pi = \left(\pi(1), \ldots, \pi(n)\right)$
		\FOR{$i=1$ \TO $n$}
		\STATE $KT(i) \gets \text{Kendall-Tau}(\sigma, \left(\pi(i), \pi(i+1), \ldots, \pi(n), \pi(1), \ldots, \pi(i-1)\right))$ 
		\ENDFOR
		\STATE $\text{best score} \gets \max_{i=1,\ldots,n} KT(i)$
		\ENSURE best score
	\end{algorithmic}
\end{algorithm}

%% file: sections/10_appendix_supp_figure.tex
\section{Additional Numerical Results}

\subsection{Genome assembly experiment (detailed)}\label{ssec:supp-genome-assembly-exp}
Here we provide background about the application of seriation methods for genome assembly and details about our experiment.
We used the {\it E. coli} reads from \citet{Loman15}. They were sequenced with Oxford Nanopore Technology (ONT) MinION device. The sequencing experiment is detailed in \url{http://lab.loman.net/2015/09/24/first-sqk-map-006-experiment} where the data is available.
The overlaps between raw reads were computed with minimap2 \citep{li2018minimap2} with the ONT preset.
The similarity matrix was constructed directly from the output of minimap2. For each pair $(i,j)$ of reads where an overlap was found, we let the number of matching bases be the similarity value associated (and zero where no overlap are found). The only preprocessing on the matrix is that we set a threshold to remove short overlaps. In practice we set the threshold to the median of the similarity values, \ie, we discard the lower half of the overlaps.
We then apply our method to the similarity matrix.
The laplacian embedding is shown in Figure~\ref{subfig:ecoli-3d-embedding}. We used no scaling of the Laplacian as it corrupted the filamentary structure of the embedding, but we normalized the similarity matrix beforehand with $W \gets D^{-1} W D^{-1}$ as in \citet{coifman2006diffusion}.
The resulting similarity matrix $S$ computed from the embedding in Algorithm~\ref{algo:Recovery_order_filamentary} is disconnected. Then, Algorithm~\ref{alg:spectral} is applied in each connected component, yielding a fragmented assembly with correctly ordered contigs, as shown in Figure~\ref{subfig:ecoli-partial-orderings}.
However, if the new similarity matrix $S$ is disconnected, the input matrix $A$ is connected. The fragmentation happened while ``scanning'' the nearest-neighbors from the embedding.
One can therefore merge the ordered contigs using the input matrix $A$ as follows.
For each contig, we check from $A$ if there are non-zero overlaps between reads at the edges of that contig and some reads at the edges of another contig. If so, we merge the two contigs, and repeat the procedure until there is only one contig left (or until there is no more overlaps between edges from any two contigs). This procedure is detailed in Algorithm~\ref{alg:connect_clusters}.
Note that the {\it E. coli}  genome is circular, therefore computing the layout should be casted as a \ref{eqn:circ-seriation} problem, as illustrated in Figure~\ref{fig:circular-genome-illustration}. Yet, since the genome is fragmented in subsequences since $S$ is disconnected, we end up using Algorithm~\ref{alg:spectral} in each connected component, \ie, solving an instance of \ref{eqn:seriation} in each contig.
\begin{figure}[hbt]
	\begin{center}
			\includegraphics[width=.4\textwidth]{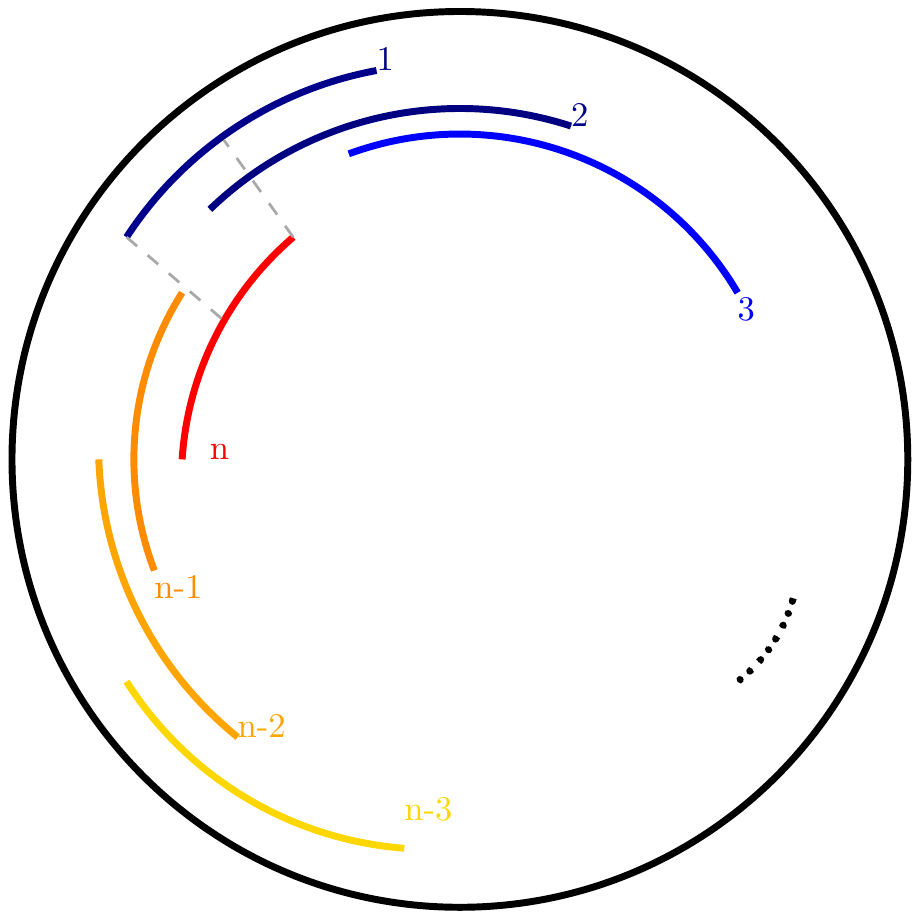}
		\caption{
			Illustration of why the overlap-based similarity matrix of an ideal circular genome should be $\circR$.
		}
		\label{fig:circular-genome-illustration}
	\end{center}
\end{figure}

\begin{figure}[hbt]
	\begin{center}
		\begin{subfigure}[htb]{0.45\textwidth}
			\includegraphics[width=\textwidth]{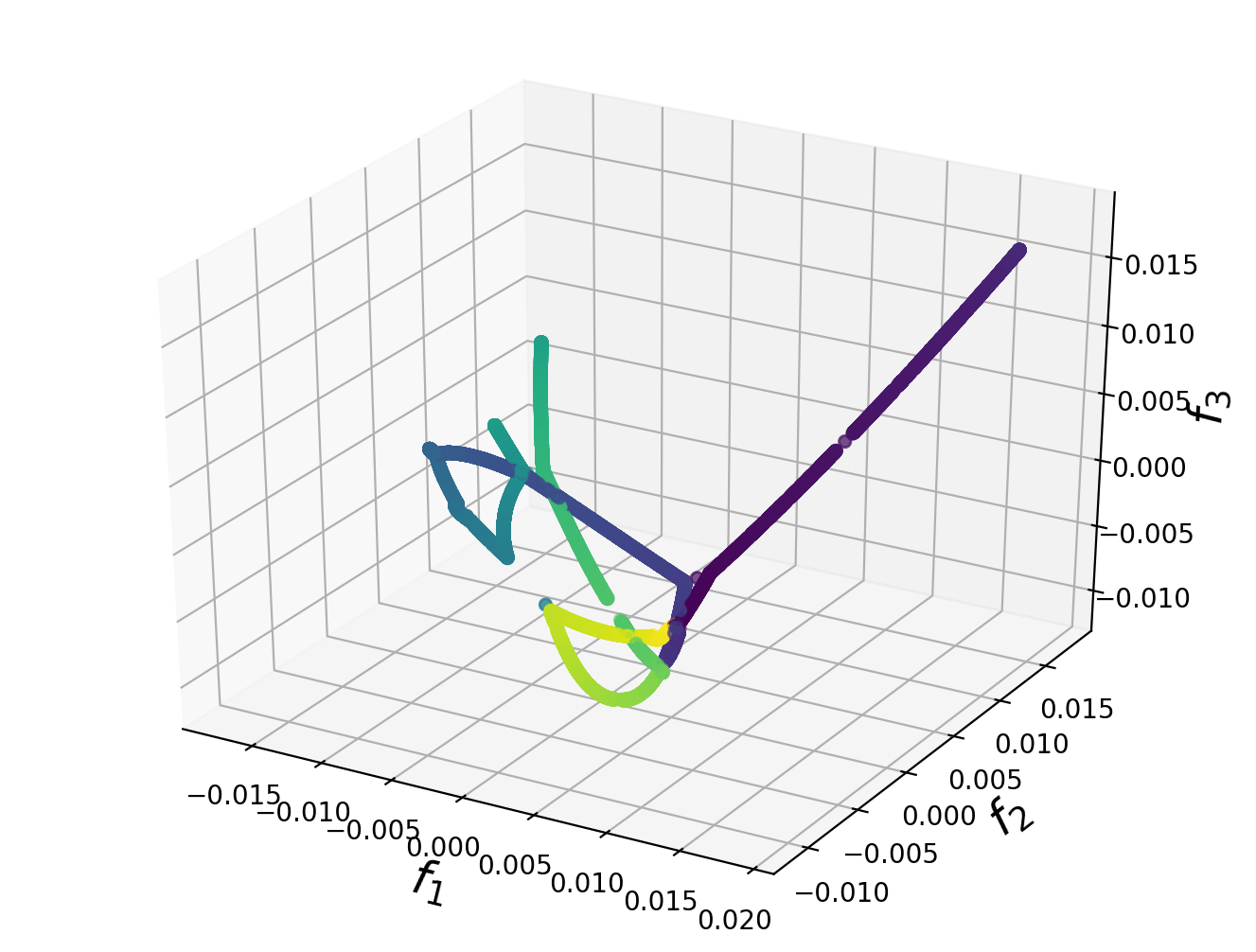}
			\caption{\kLE{3}}\label{subfig:ecoli-3d-embedding}
		\end{subfigure}
		\begin{subfigure}[htb]{0.45\textwidth}
			\includegraphics[width=\textwidth]{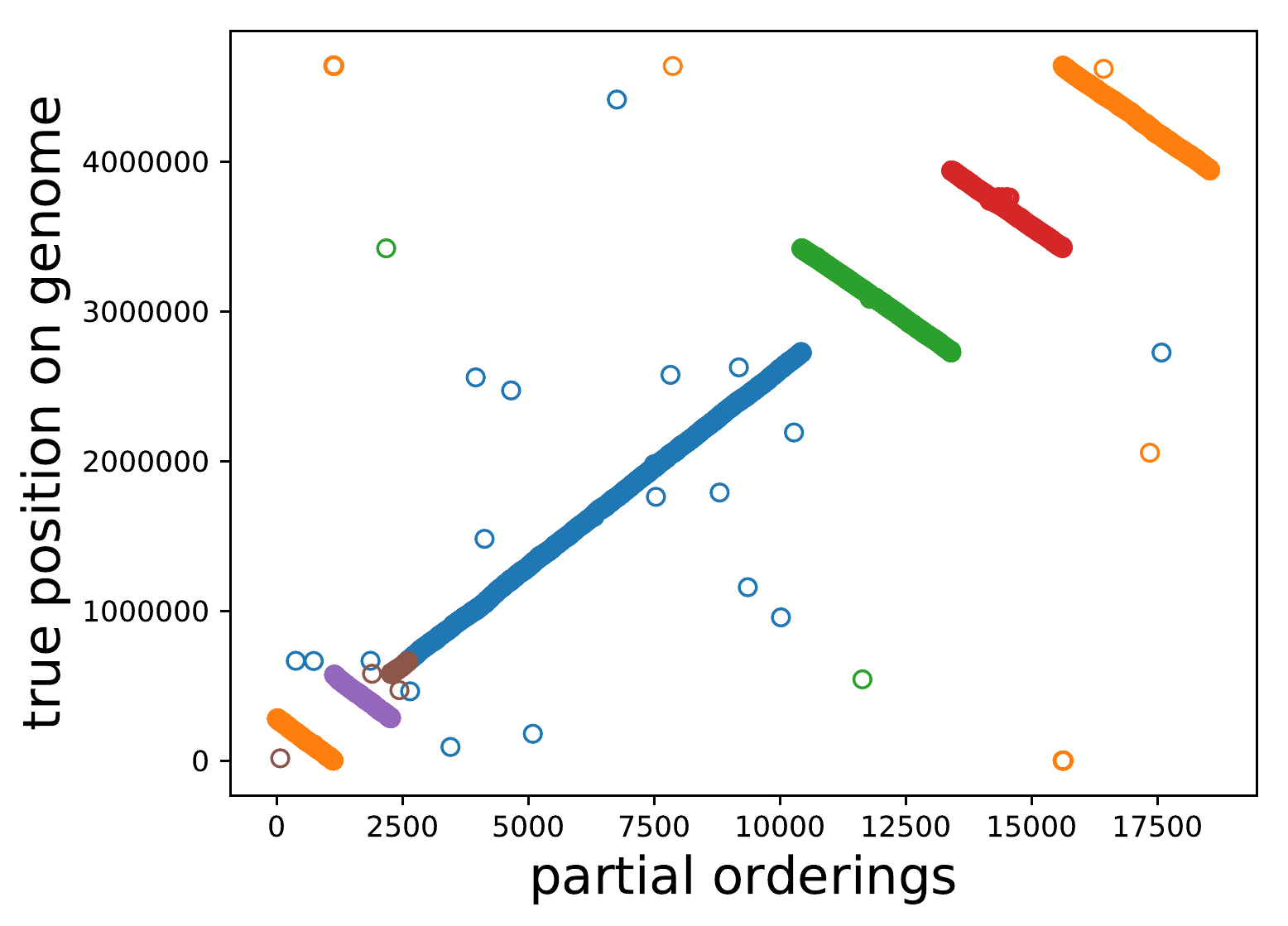}
			\caption{partial orderings}\label{subfig:ecoli-partial-orderings}
		\end{subfigure}
		\caption{
			3d Laplacian embedding from {\it E. coli} reads overlap-based similarity matrix (\ref{subfig:ecoli-3d-embedding}), and the orderings found in each connected component of the new similarity matrix created in Algorithm~\ref{algo:Recovery_order_filamentary} (\ref{subfig:ecoli-partial-orderings}) versus the position of the reads within a reference genome obtained by mapping tge reads to the reference with minimap2 (all plotted on the same plot for compactness).
			The orderings have no absolute direction, \ie, $(1,2,\ldots,n)$ and $(n,n-1,\ldots,1)$ are equivalent, which is why the lines in subfigure~\ref{subfig:ecoli-partial-orderings} can be either diagonal or anti-diagonal.
		}
		\label{fig:ecoli-exp-supp}
	\end{center}
	\vskip -0.2in
\end{figure}

The experiment can be reproduced with the material on \url{https://github.com/antrec/mdso}, and the parameters easily varied.
Overall, the final ordering found is correct when the threshold on the overlap-based similarity is sufficient (in practice, above $\sim40\%$ of the non-zero values).
When the threshold increases or when the number of nearest neighbors $k$ from Algorithm~\ref{algo:Recovery_order_filamentary} decreases, the new similarity matrix $S$ gets more fragmented, but the final ordering remains the same after the merging procedure.

\subsection{Gain over baseline}
In Figure~\ref{fig:exp-main-banded}, each curve is the mean of the Kendall-tau (a score directly interpretable by practitioners) over many different Gaussian random realizations of the noise. The shaded confidence interval represents the area in which the true expectation is to be with high probability but not the area in which the score of an experiment with a given noisy similarity would be. As mentioned in the main text, the shaded interval is the standard deviation divided by $\sqrt{n_{\text{exps}}}$, since otherwise the plot was hard to read, as the intervals crossed each others.

Practitioners may use this method in one-shot (e.g. for one particular data-set). In that case, it would be more relevant to show directly the standard deviation on the plots, which is the same as what is displayed, but multiplied by 10. Then, the confidence intervals between the baseline and our method would cross each other. However, the standard deviation on all experiments is due to the fact that some instances are more difficult to solve than some others. On the difficult instances, the baseline and our method perform more poorly than on easy instances.
However, we  also computed the gain over the baseline, \ie, the difference of score between our method and the baseline, for each experiment, and it is always, or almost always positive, \ie, our method almost always beats the baseline although the confidence intervals cross each other.

\subsection{Numerical results with KMS matrices}\label{ssec:app-KMS-main-exp}
In Figure~\ref{fig:exp-main-KMS} we show the same plots as in Section~\ref{sec:numerical_results} but with matrices $A$ such that $A_{ij} = e^{\alpha |i-j|}$, with $\alpha=0.1$ and $n=500$.

\begin{figure}[hbt]
	\begin{center}
		\begin{subfigure}[htb]{0.45\textwidth}
			\includegraphics[width=\textwidth]{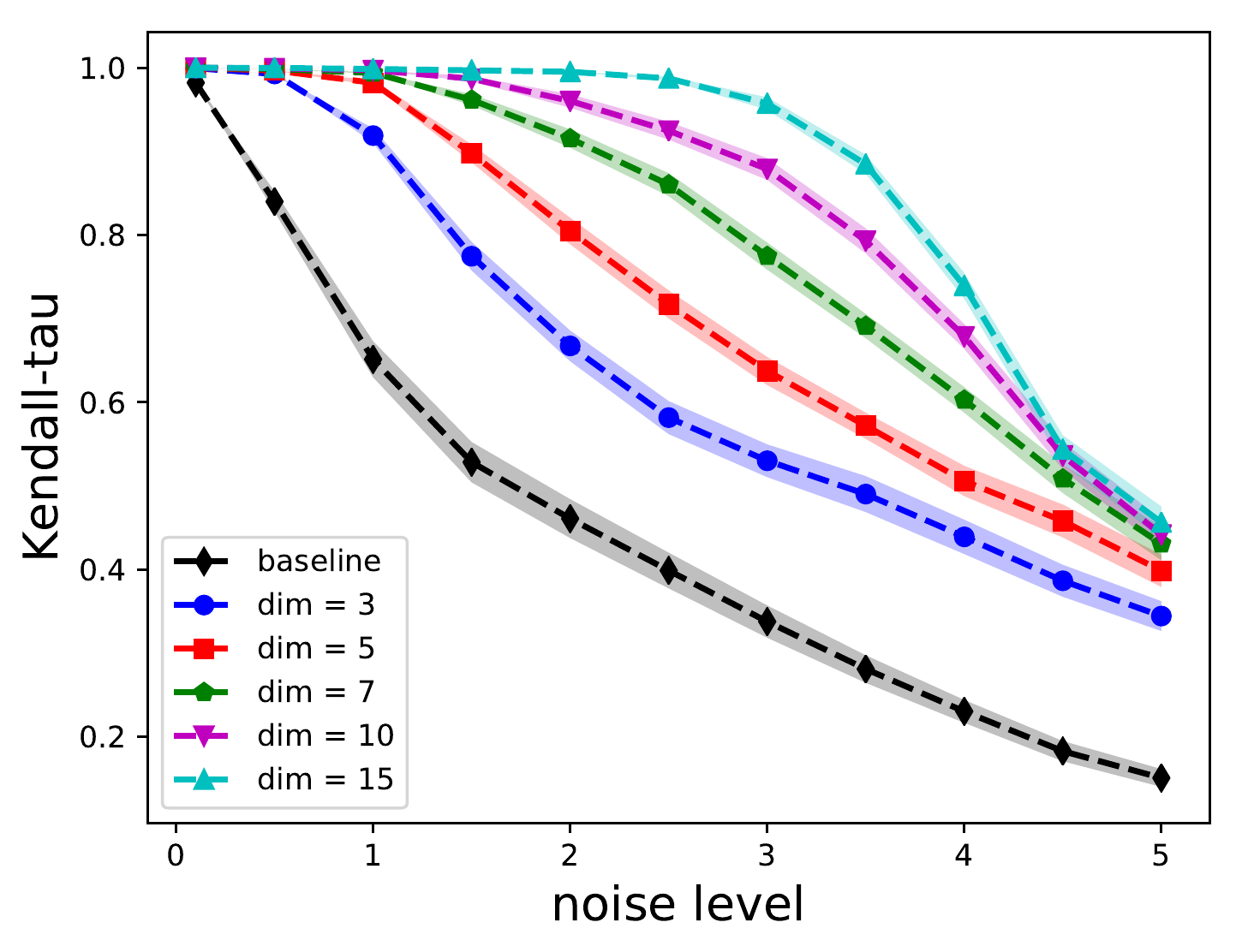}
			\caption{Linear KMS}\label{subfig:exps-lin-KMS}
		\end{subfigure}
		\begin{subfigure}[htb]{0.45\textwidth}
			\includegraphics[width=\textwidth]{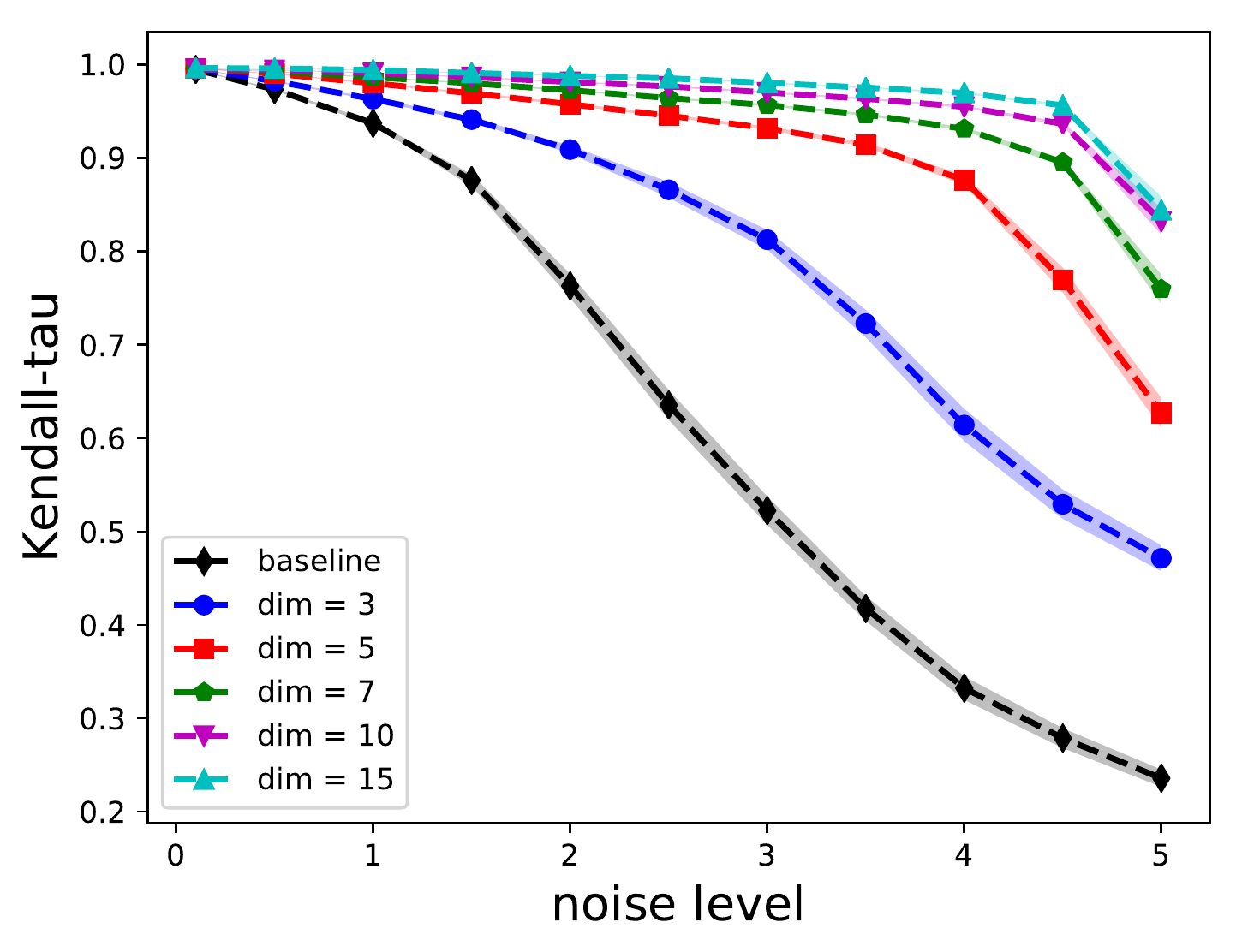}
			\caption{Circular KMS}\label{subfig:exps-circ-KMS}
		\end{subfigure}
		\caption{
		K-T scores for Linear (\ref{subfig:exps-lin-KMS}) and Circular (\ref{subfig:exps-circ-KMS}) Seriation for noisy observations of KMS, Toeplitz, matrices, displayed for several values of the dimension parameter of the \kLE{d}.
		}
		\label{fig:exp-main-KMS}
	\end{center}
	\vskip -0.2in
\end{figure}

\subsection{Sensitivity to parameter $k$ (number of neighbors)}\label{ssec:app-k-sensitivity}
Here we show how our method performs when we vary the parameter $k$ (number of neighbors at step 4 of Algorithm~\ref{algo:Recovery_order_filamentary}), for both linearly decrasing, banded matrices, $A_{ij} = \max \left( c - |i-j|, 0, \right)$ (as in Section~\ref{sec:numerical_results}), in Figure~\ref{fig:exp-ksensitivity-banded} and with matrices $A$ such that $A_{ij} = e^{\alpha |i-j|}$, with $\alpha=0.1$ (Figure~\ref{fig:exp-ksensitivity-KMS}.

\begin{figure}[hbt]
	\begin{center}
		\begin{subfigure}[htb]{0.45\textwidth}
			\includegraphics[width=\textwidth]{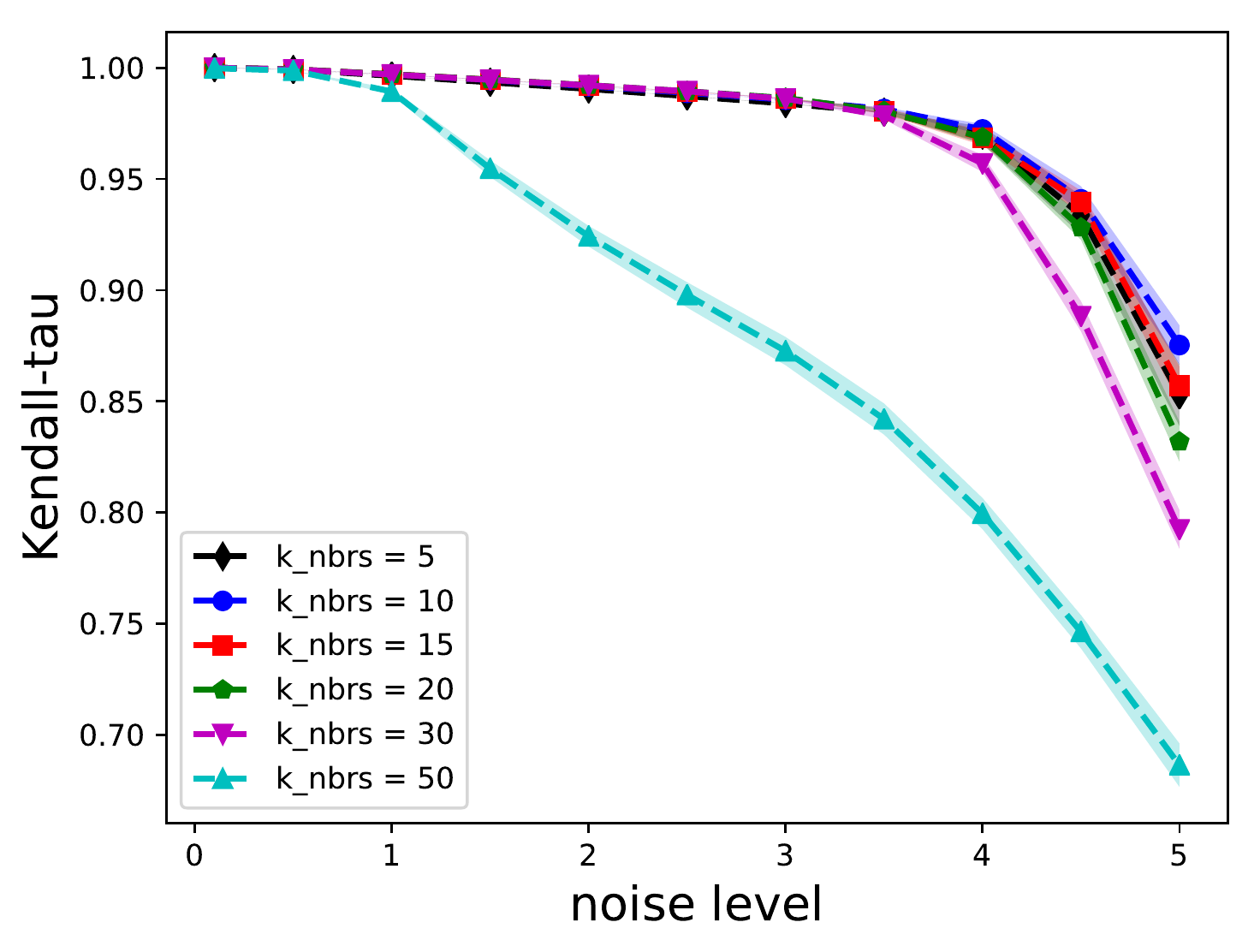}
			\caption{Linear Banded}\label{subfig:exps-lin-banded-ksensitivity}
		\end{subfigure}
		\begin{subfigure}[htb]{0.45\textwidth}
			\includegraphics[width=\textwidth]{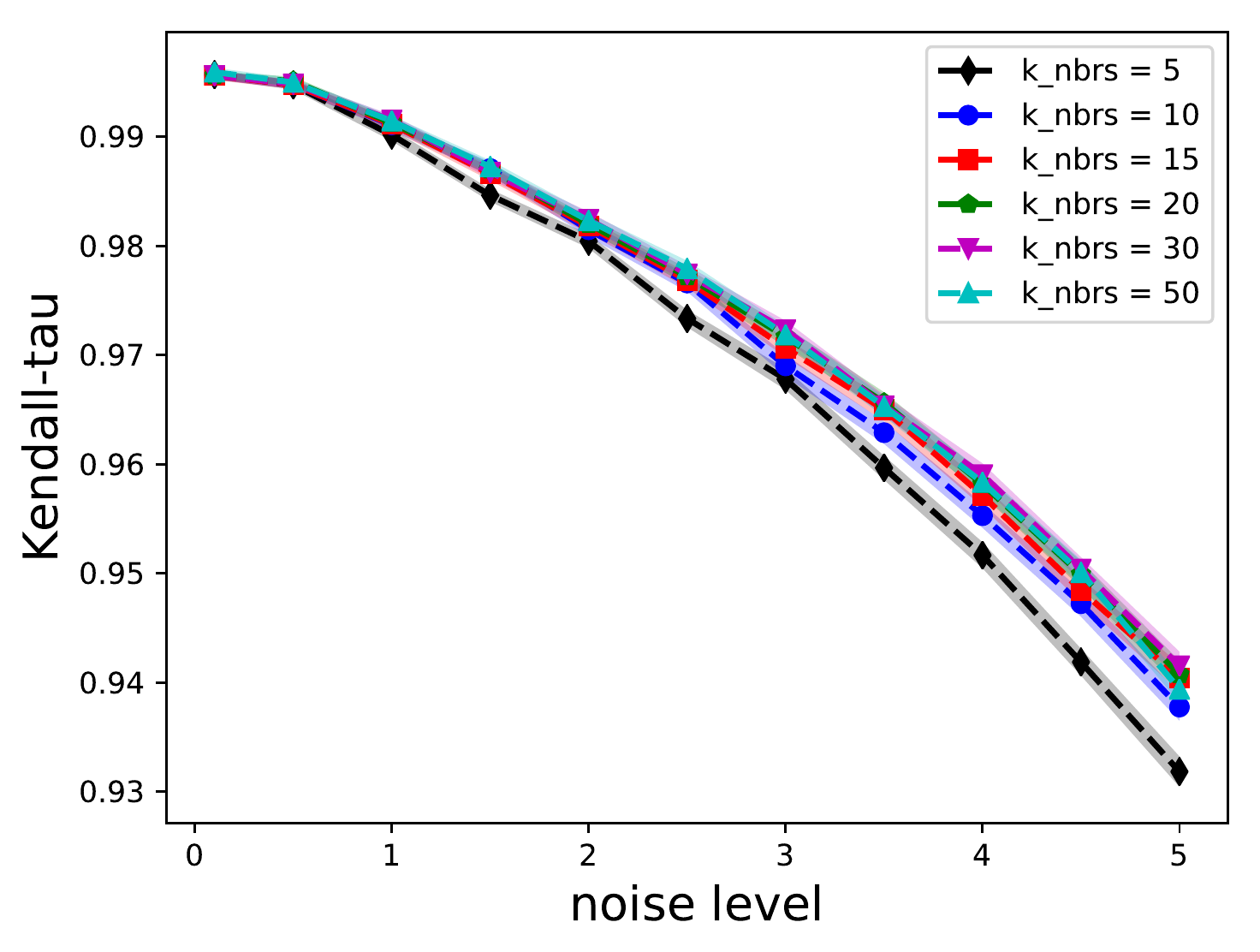}
			\caption{Circular Banded}\label{subfig:exps-circ-banded-ksensitivity}
		\end{subfigure}
		\caption{
		K-T scores for Linear (\ref{subfig:exps-lin-banded-ksensitivity}) and Circular (\ref{subfig:exps-circ-banded-ksensitivity}) Seriation for noisy observations of banded, Toeplitz, matrices, displayed for several values of the number of nearest neighbors $k$, with a fixed value of the dimension of the \kLE{d}, $d=10$.
		}
		\label{fig:exp-ksensitivity-banded}
	\end{center}
	\vskip -0.2in
\end{figure}

\begin{figure}[hbt]
	\begin{center}
		\begin{subfigure}[htb]{0.45\textwidth}
			\includegraphics[width=\textwidth]{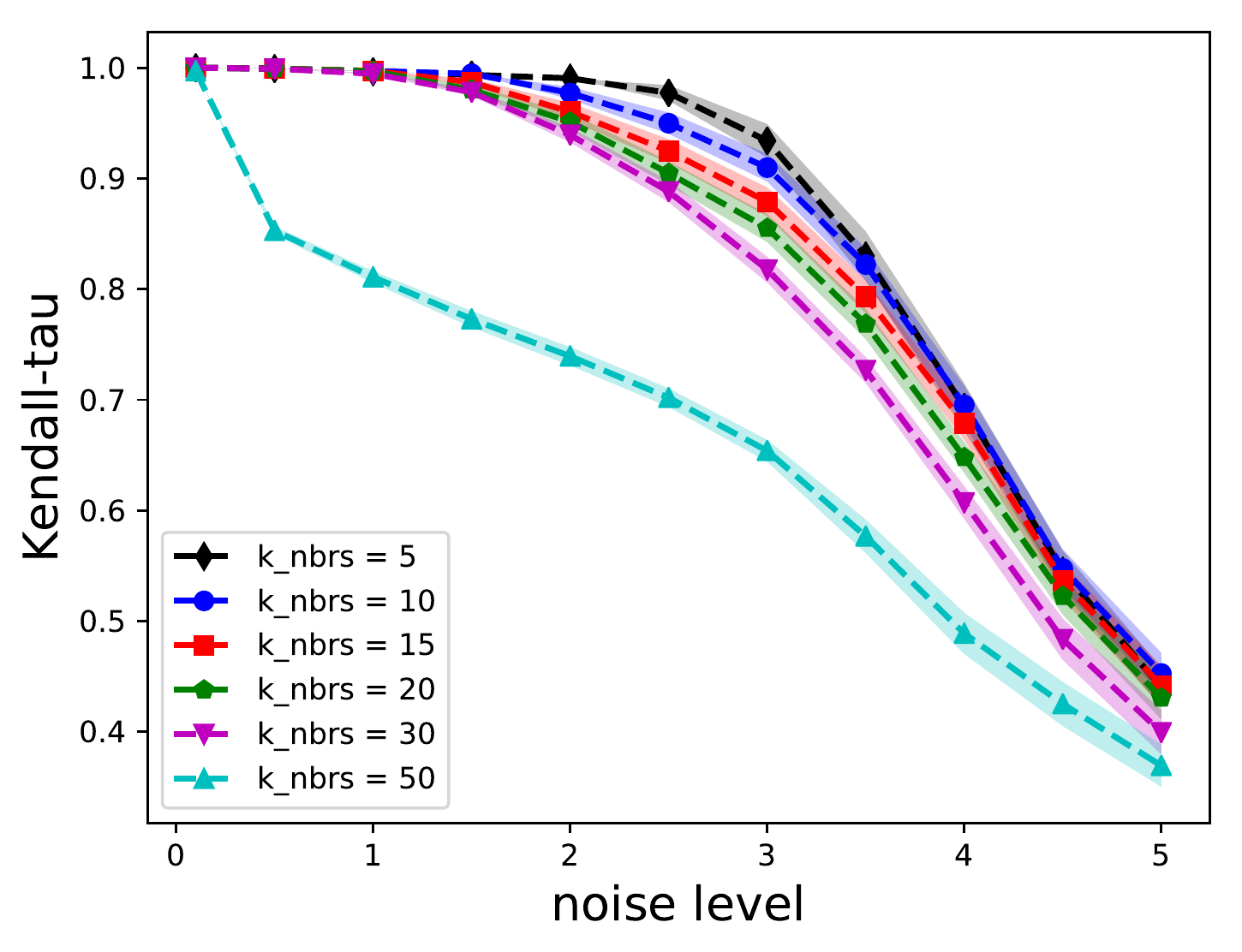}
			\caption{Linear KMS}\label{subfig:exps-lin-KMS-ksensitivity}
		\end{subfigure}
		\begin{subfigure}[htb]{0.45\textwidth}
			\includegraphics[width=\textwidth]{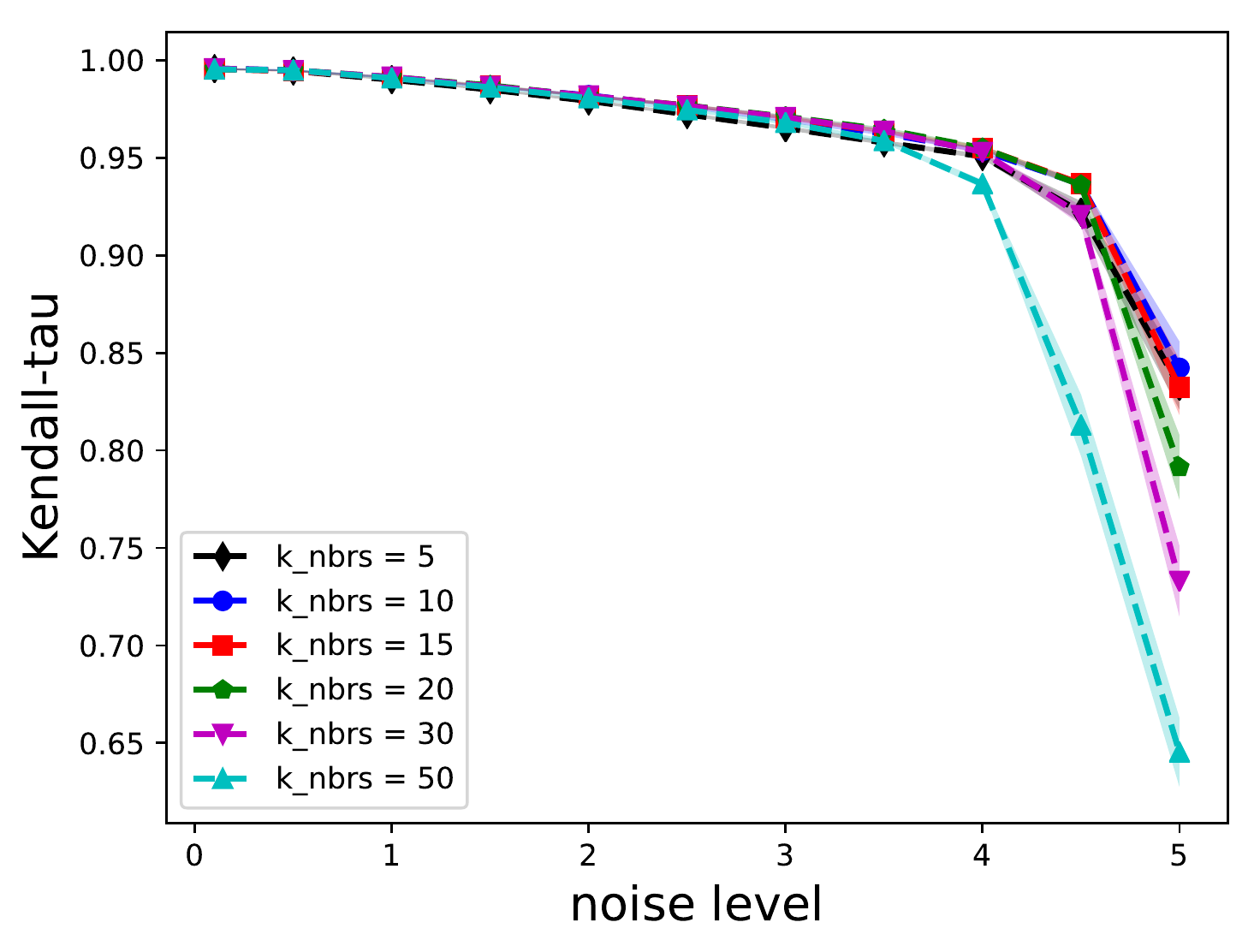}
			\caption{Circular KMS}\label{subfig:exps-circ-KMS-ksensitivity}
		\end{subfigure}
		\caption{
		K-T scores for Linear (\ref{subfig:exps-lin-KMS-ksensitivity}) and Circular (\ref{subfig:exps-circ-KMS-ksensitivity}) Seriation for noisy observations of KMS, Toeplitz, matrices, displayed for several values of the number of nearest neighbors $k$, with a fixed value of the dimension of the \kLE{d}, $d=10$.
		}
		\label{fig:exp-ksensitivity-KMS}
	\end{center}
	\vskip -0.2in
\end{figure}

We observe that the method performs roughly equally well with $k$ in a range from 5 to 20, and that the performances drop when $k$ gets too large, around $k=30$. This can be interpreted as follows. When $k$ is too large, the assumption that the points in the embedding are locally fitted by a line no longer holds.
Note also that in practice, for small values of $k$, \eg, $k=5$, the new similarity matrix $S$ can be disconnected, and we have to resort to the merging procedure described in Algorithm~\ref{alg:connect_clusters}.

\subsection{Sensitivity to the normalization of the Laplacian}\label{ssec:scaling-sensitivity}
We performed experiments to compare the performances of the method with the default Laplacian embedding \eqref{eqn:kLE} (red curve in Figure~\ref{fig:exp-appendix-scaling} and~\ref{fig:exp-appendix-scaling-20}) and with two possible normalized embeddings \eqref{eqn:scaled-kLE} (blue and black curve).
We observed that with the default \ref{eqn:kLE}, the performance first increases with $d$, and then collapses when $d$ gets too large.
The CTD scaling (blue) has the same issue, as the first $d$ eigenvalues are roughly of the same magnitude in our settings.
The heuristic scaling \ref{eqn:scaled-kLE} with $\alpha_k = 1/\sqrt{k}$ that damps the higher dimensions yields better results when $d$ increases, with a plateau rather than a collapse when $d$ gets large.
We interpret these results as follows.
With the  \eqref{eqn:kLE}, Algorithm \ref{algo:Recovery_order_filamentary}, line \ref{line:find_direction}  treats equally all dimensions of the embedding. However, the curvature of the embedding tends to increase with the dimension (for $\circR$ matrix, the period of the cosines increases linearly with the dimension). The filamentary structure is less smooth and hence more sensitive to noise in high dimensions, which is why the results are improved by damping the high dimensions (or using a reasonably small value for $d$).

\begin{figure}[hbt]
	\begin{center}
		\begin{subfigure}[htb]{0.45\textwidth}
			\includegraphics[width=\textwidth]{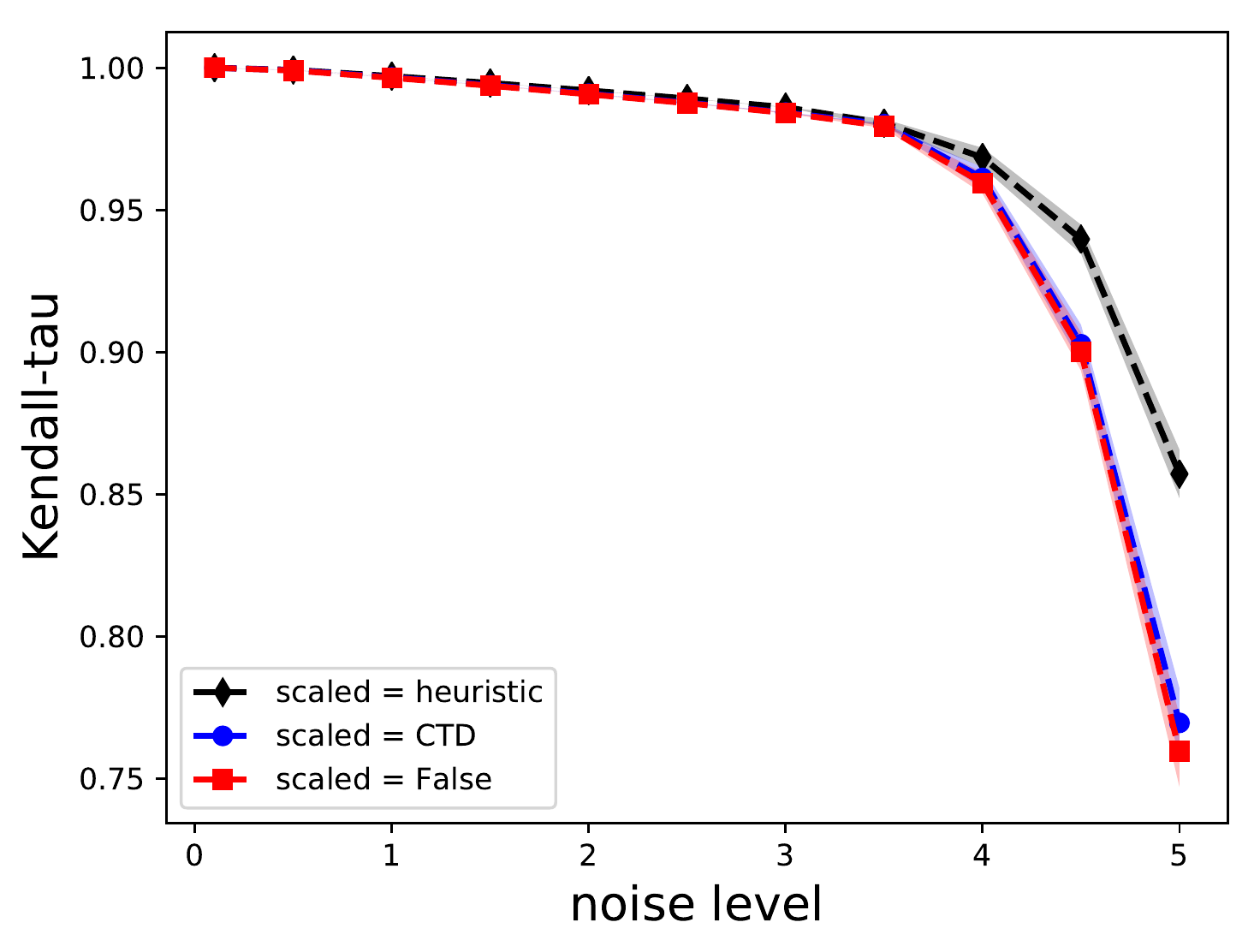}
			\caption{Linear Banded}\label{subfig:exps-lin-banded-scaling}
		\end{subfigure}
		\begin{subfigure}[htb]{0.45\textwidth}
			\includegraphics[width=\textwidth]{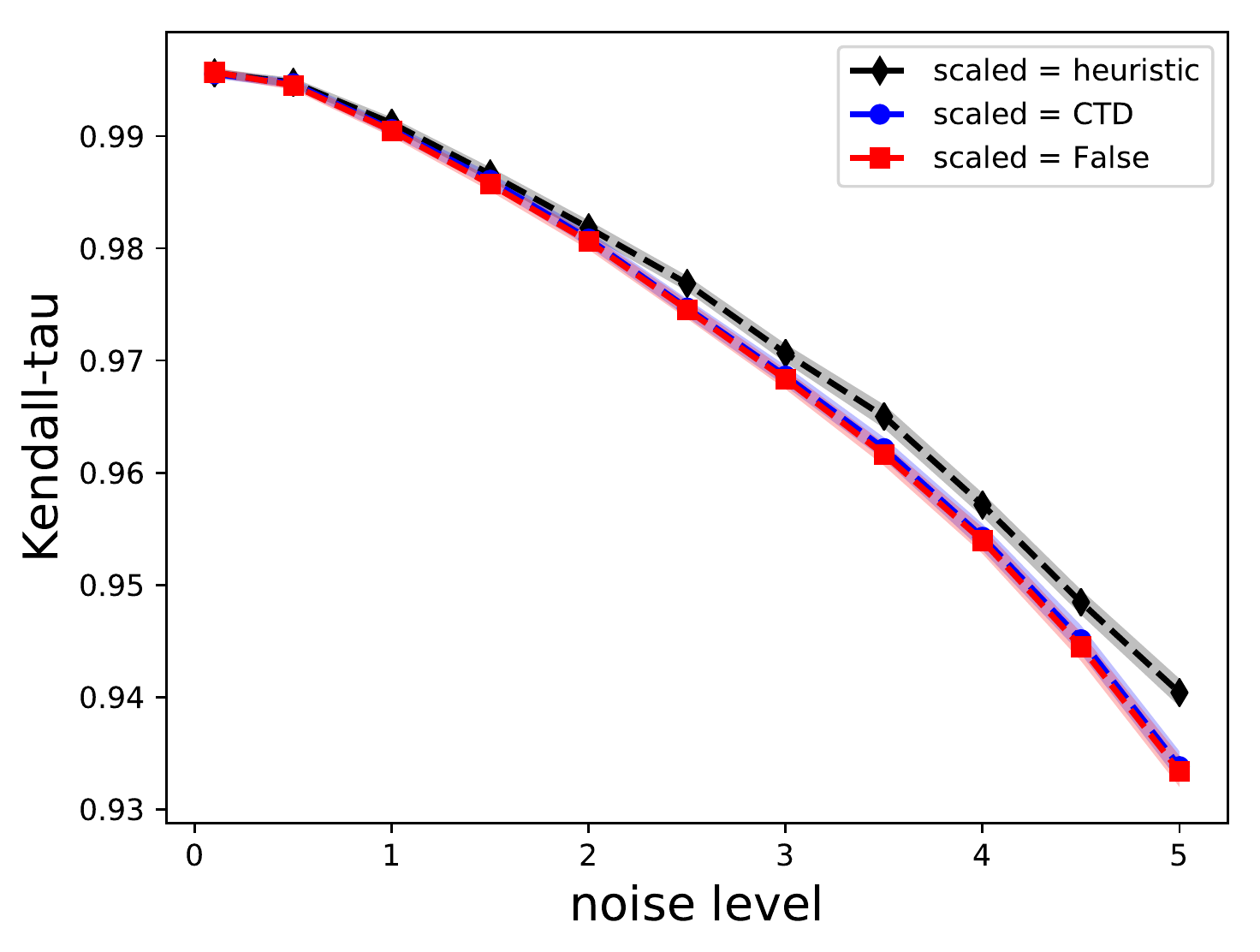}
			\caption{Circular Banded}\label{subfig:exps-circ-banded-scaling}
		\end{subfigure}
		\caption{
			Mean of Kendall-Tau for Linear (\ref{subfig:exps-lin-banded-scaling}) and Circular (\ref{subfig:exps-circ-banded-scaling}) Seriation for noisy observations of banded, Toeplitz, matrices, displayed for several scalings of the Laplacian embedding, with a fixed number of neighbors $k=15$ and number of dimensions $d=10$ in the \kLE{d}.}
		\label{fig:exp-appendix-scaling}
	\end{center}
	\vskip -0.2in
\end{figure}
\begin{figure}[hbt]
	\begin{center}
		\begin{subfigure}[htb]{0.45\textwidth}
			\includegraphics[width=\textwidth]{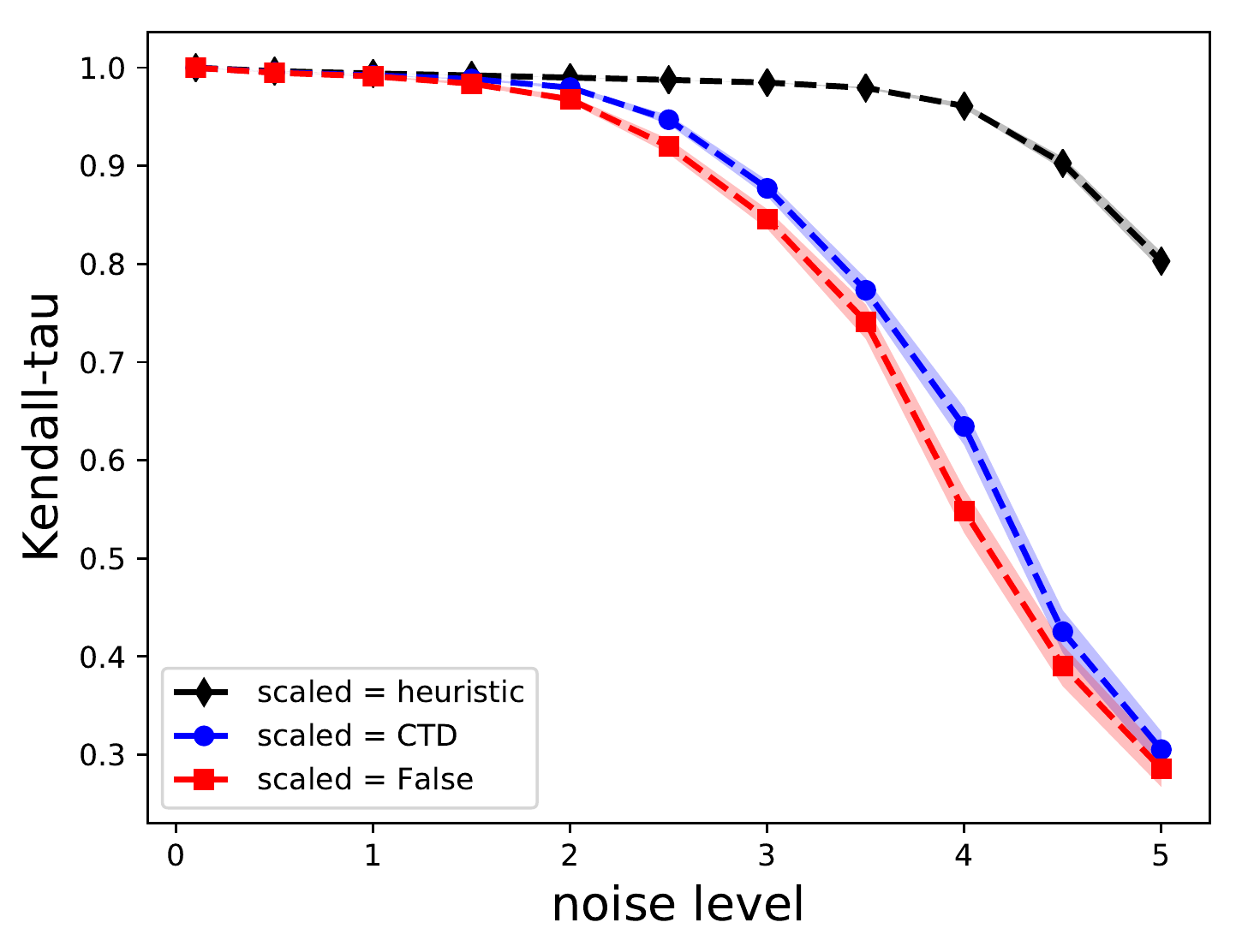}
			\caption{Linear Banded}\label{subfig:exps-lin-banded-scaling-20}
		\end{subfigure}
		\begin{subfigure}[htb]{0.45\textwidth}
			\includegraphics[width=\textwidth]{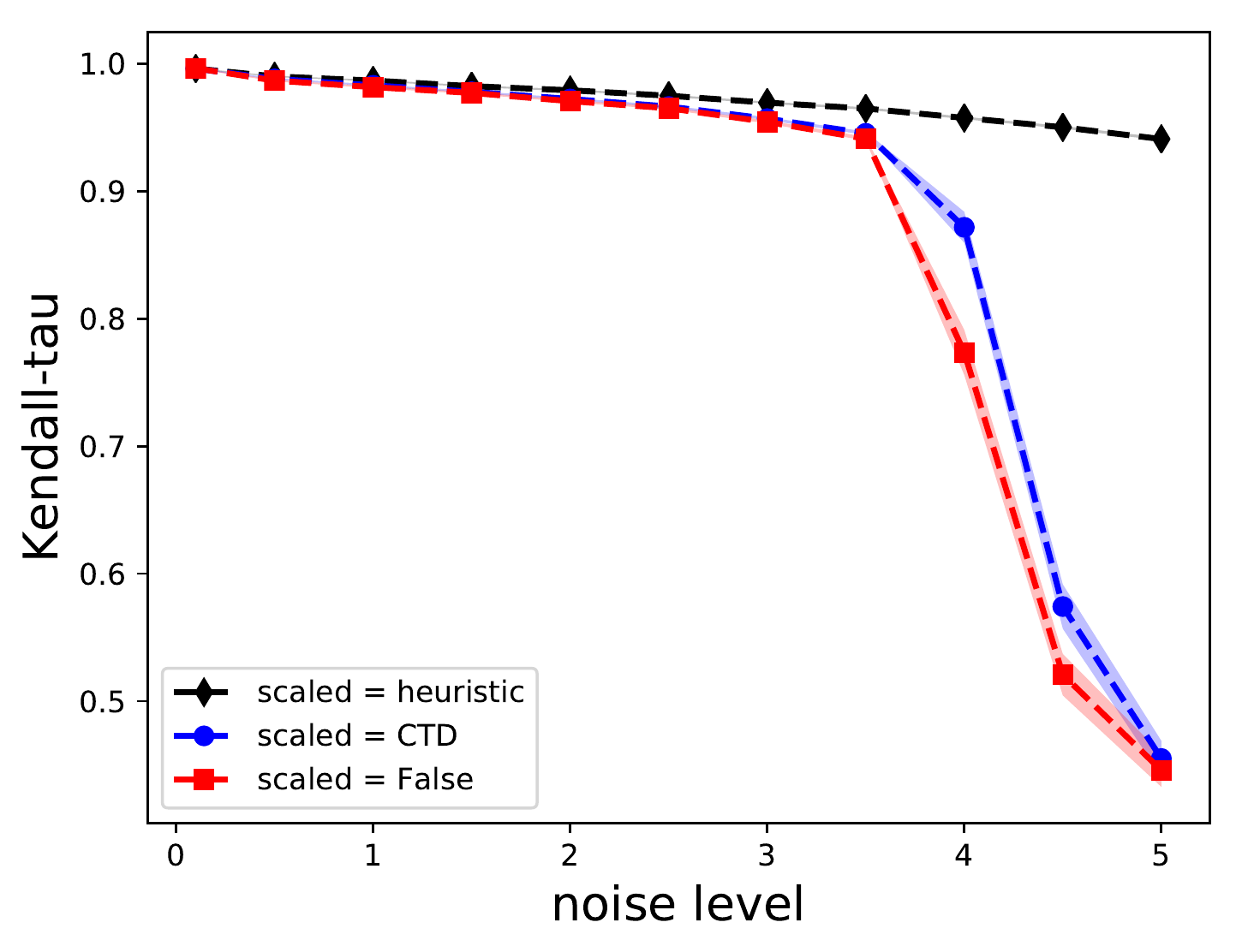}
			\caption{Circular Banded}\label{subfig:exps-circ-banded-scaling-20}
		\end{subfigure}
		\caption{
			Mean of Kendall-Tau for Linear (\ref{subfig:exps-lin-banded-scaling-20}) and Circular (\ref{subfig:exps-circ-banded-scaling-20}) Seriation for noisy observations of banded, Toeplitz, matrices, displayed for several scalings of the Laplacian embedding, with a fixed number of neighbors $k=15$ and number of dimensions $d=20$ in the \kLE{d}.}
		\label{fig:exp-appendix-scaling-20}
	\end{center}
	\vskip -0.2in
\end{figure}

\subsection{Illustration of Algorithm~\ref{algo:Recovery_order_filamentary}}\label{ssec:illustrations}
Here we provide some visual illustrations of the method with a circular banded matrix.
Given a matrix $A$ (Figure~\ref{subfig:exps-matshow-noisy}), Algorithm~\ref{algo:Recovery_order_filamentary} computes the \kLE{d}. The \kLE{2} is plotted for visualization in Figure~\ref{subfig:exps-embedding-noisy}.
Then, it creates a new matrix $S$ (Figure~\ref{subfig:exps-matshow-clean}) from the local alignment of the points in the \kLE{d}.
Finally, from the new matrix $S$, it computes the \kLE{2} (Figure~\ref{subfig:exps-matshow-clean}), on which it runs the simple method from Algorithm~\ref{alg:circular-2d-ordering}.

Figure~\ref{fig:exp-noisy-illustration} and \ref{fig:exp-clean-illustration} give a qualitative illustration of how the method behaves compared to the basic Algorithm~\ref{alg:circular-2d-ordering}.

\begin{figure}[hbt]
	\begin{center}
		\begin{subfigure}[htb]{0.35\textwidth}
			\includegraphics[width=\textwidth]{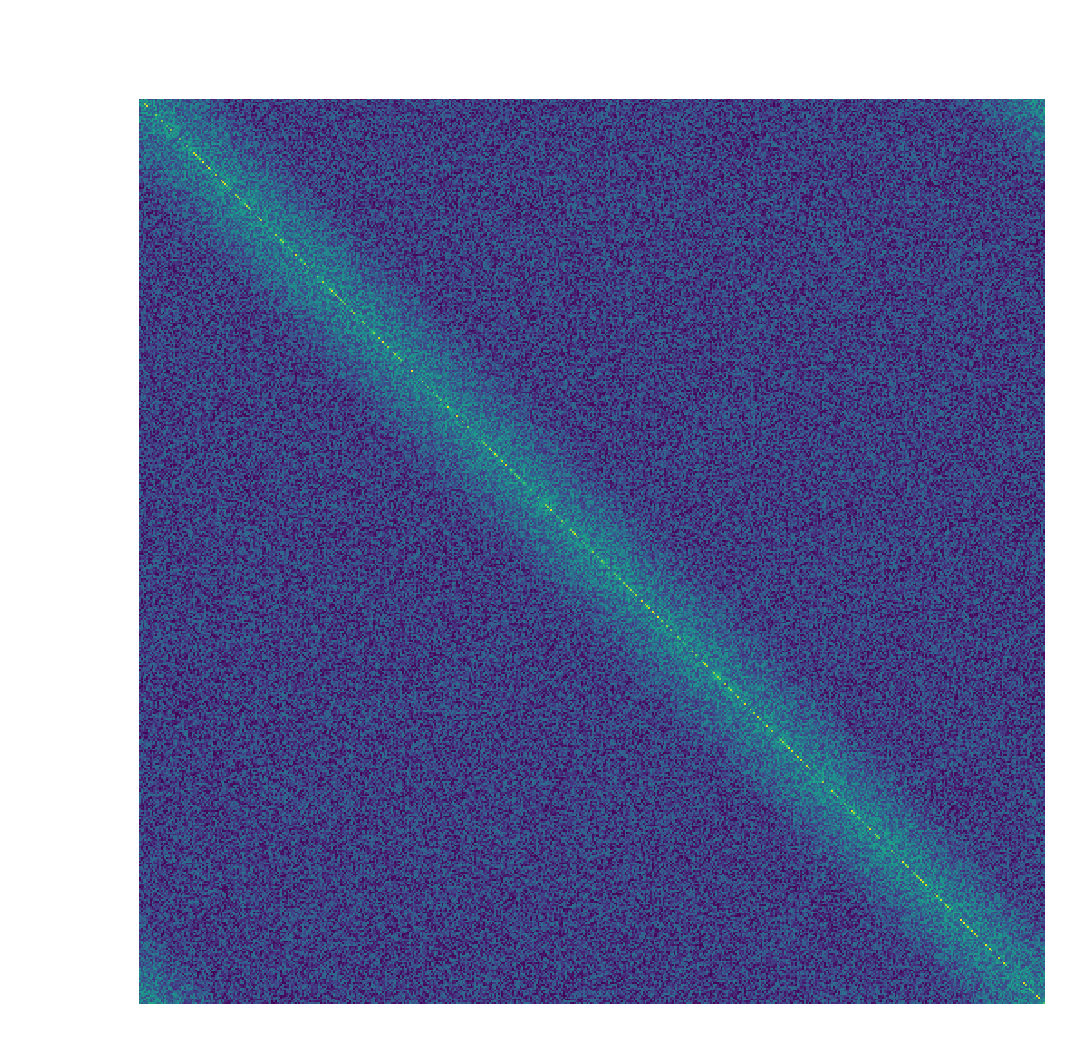}
			\caption{Noisy circular banded matrix $A$}\label{subfig:exps-matshow-noisy}
		\end{subfigure}
		\begin{subfigure}[htb]{0.45\textwidth}
			\includegraphics[width=\textwidth]{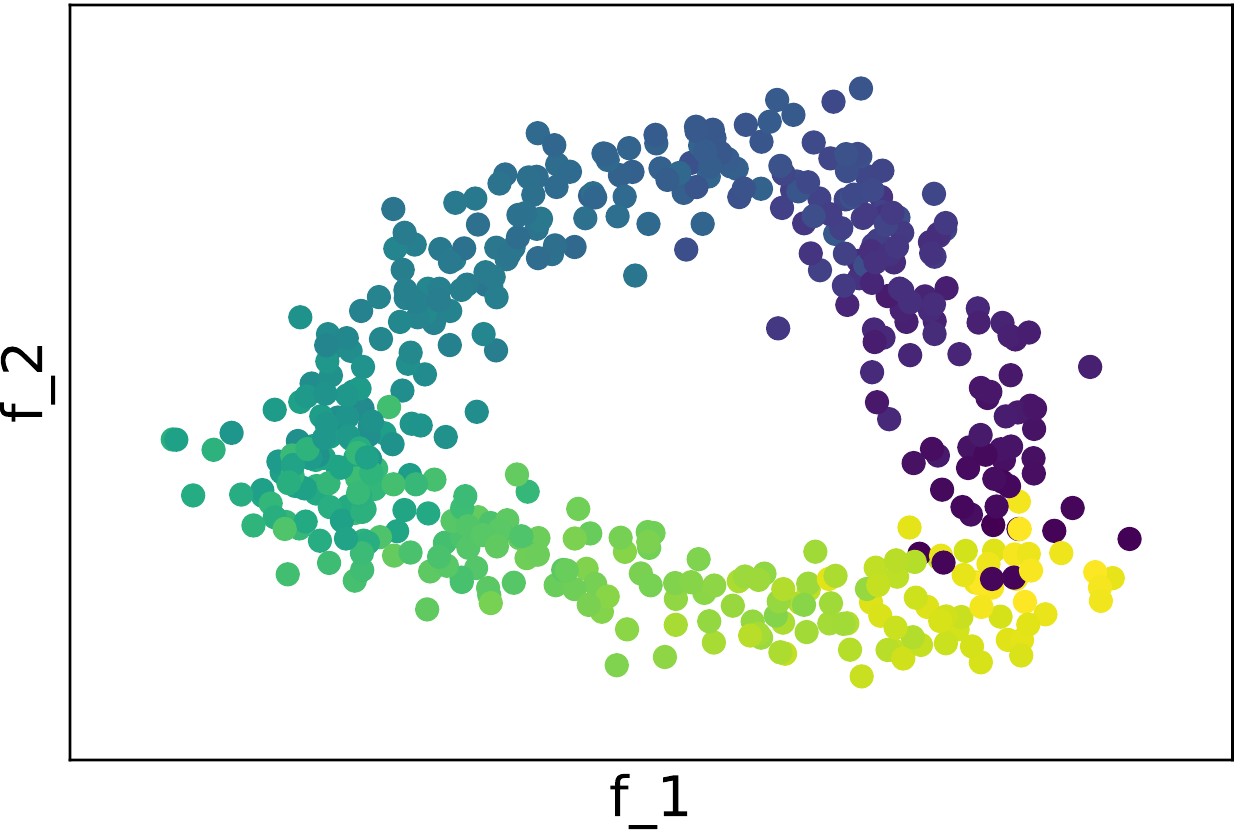}
			\caption{Noisy \kLE{2}}\label{subfig:exps-embedding-noisy}
		\end{subfigure}
		\caption{
		Noisy Circular Banded matrix (\ref{subfig:exps-matshow-noisy}) and associated 2d Laplacian embedding (\ref{subfig:exps-embedding-noisy}).
		}
		\label{fig:exp-noisy-illustration}
	\end{center}
	\vskip -0.2in
\end{figure}

\begin{figure}[hbt]
	\begin{center}
		\begin{subfigure}[htb]{0.35\textwidth}
			\includegraphics[width=\textwidth]{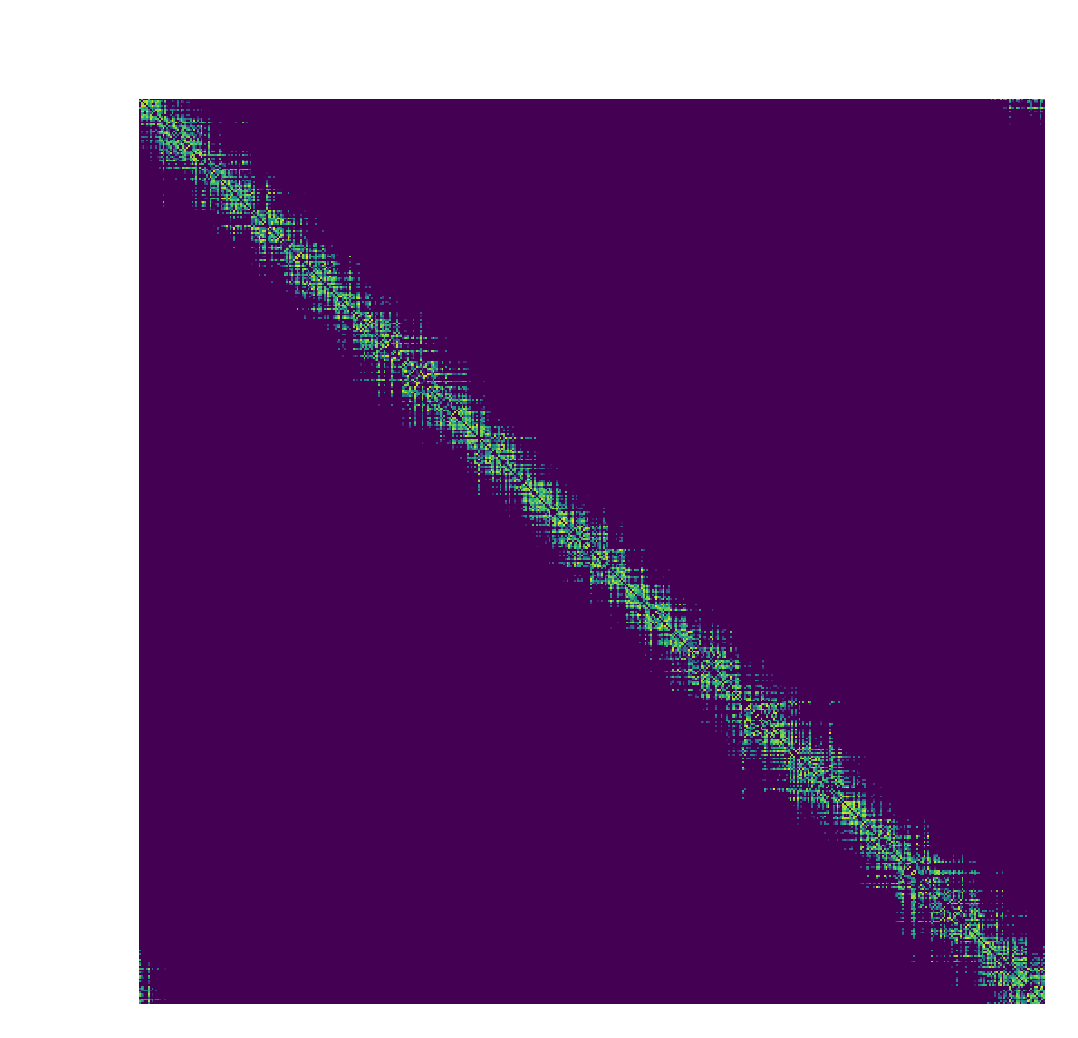}
			\caption{Matrix $S$ from Algorithm~\ref{algo:Recovery_order_filamentary}}\label{subfig:exps-matshow-clean}
		\end{subfigure}
		\begin{subfigure}[htb]{0.45\textwidth}
			\includegraphics[width=\textwidth]{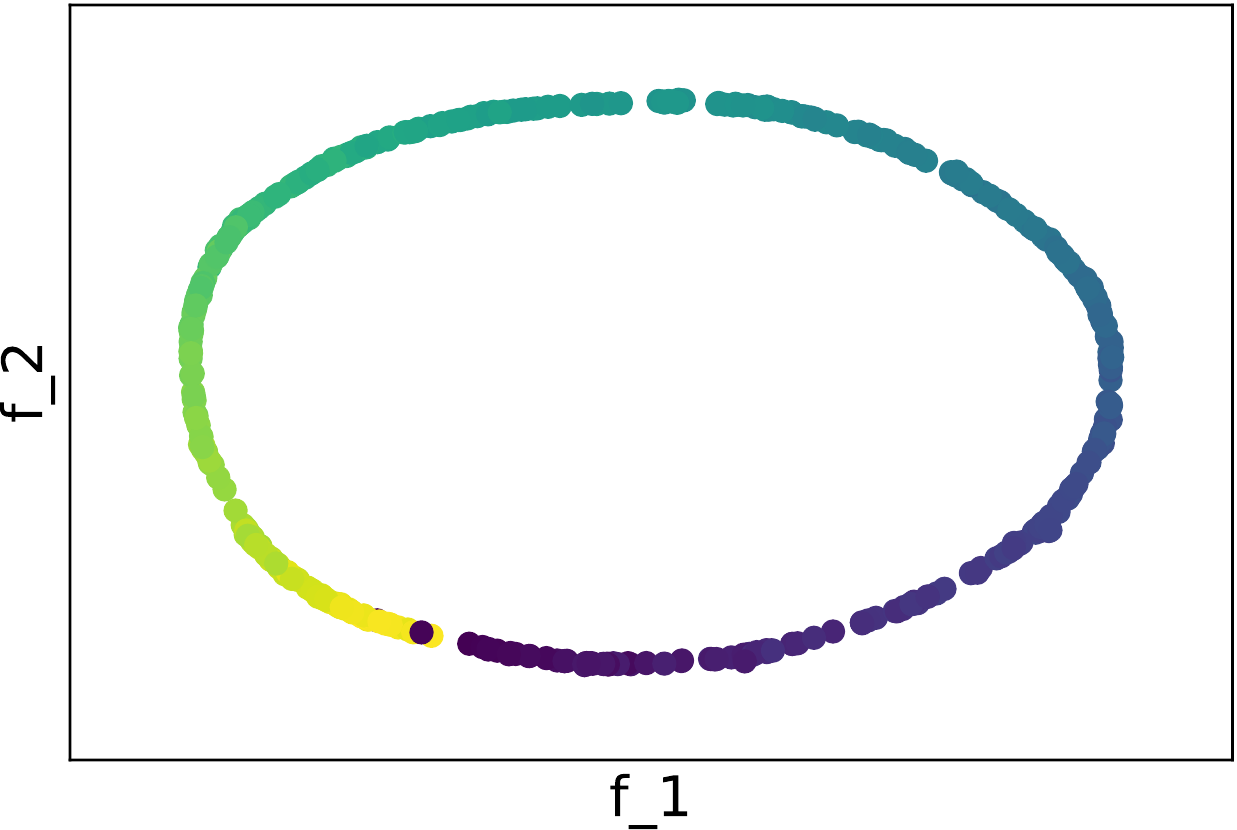}
			\caption{New \kLE{2}}\label{subfig:exps-embedding-clean}
		\end{subfigure}
		\caption{
		Matrix $S$ created through Algorithm~\ref{algo:Recovery_order_filamentary} (\ref{subfig:exps-matshow-clean}), and associated 2d-Laplacian embedding (\ref{subfig:exps-embedding-clean}).
		}
		\label{fig:exp-clean-illustration}
	\end{center}
	\vskip -0.2in
\end{figure}

%% file: sections/5_appendix_circular_proof.tex
\section{Proof of Theorem~\ref{th:without_noise}}\label{sec:circular_toeplitz_matrix}

In this Section, we prove Theorem~\ref{th:without_noise}.
There are many technical details, notably the distinction between the cases $n$ even and odd.
The key idea is to compare the sums involved in the eigenvalues of the circulant matrices $A \in \SCR$. It is the sum of the $b_k$ times values of cosines. For $\lambda_1$, we roughly have a reordering inequality where the ordering of the $b_k$ matches those of the cosines.
For the following eigenvalues, the set of values taken by the cosines is roughly the same, but it does not match the ordering of the $b_k$.
Finally, the eigenvectors of the Laplacian of $A$ are the same than those of $A$ for circulant matrices $A$, as observed in \S\ref{ssec:spectral-prop-lapl}.

We now introduce a few lemmas that will be useful in the proof.

\textbf{Notation.} In the following we denote $z_k^{(m)} \triangleq \cos(2\pi k m /n)$ and $S_p^{(m)} \triangleq \sum_{k=1}^{p}{z_k^{(m)}}$. Let's define $\mathcal{Z}_n = \{\cos(2\pi k/n)~|~k\in\mathbb{N}\}\setminus\{-1;1\}$. Depending on the parity of $n$, we will write $n=2p$ or $n=2p+1$. Hence we always have $p=\big\lfloor\frac{n}{2}\big\rfloor$. Also when $m$ and $n$ are not coprime we will note $m=d m^{\prime}$ as well as $n = d n^{\prime}$ with $n^{\prime}$ and $m^{\prime}$ coprime.

\subsection{Properties of sum of cosinus.}
The following lemma gives us how the partial sum sequence $(S_q^{(m)})$ behave for $q=p$ or $q=p-1$ as well as it proves its symmetric behavior in \eqref{eqn:sum-zk-sym-ineq}.
\begin{lemma}\label{lemma:equality_ending_partial_sum}
For $z_k^{(m)}=\cos(\frac{2\pi k m}{n})$, $n=2p+1$ and any $m=1,\ldots,p$
\begin{eqnarray}\label{eqn:sum-total-zk}
S^{(m)}_{p} \triangleq \sum_{k=1}^{p}{z_k^{(m)}} = -\frac{1}{2}~.
\end{eqnarray}
Also, for $1 \leq q \leq p/2$,
\begin{align}\label{eqn:sum-zk-sym-ineq}
S^{(1)}_{p-q} \geq S^{(1)}_{q}~.
\end{align}
For $n$ and $m\geq 2$ even ($n = 2p$), we have
\begin{align}
S^{(1)}_{p-1-q} &= S^{(1)}_{q}~~\text{for}~~ 1 \leq q \leq (p-1)/2\label{eq:sum_zk_n_m_even}\\
S^{(1)}_{p-1} &= 0 ~~\text{and}~~ S^{(m)}_{p-1} = -1 ~.\label{eq:S_m_p_1}
\end{align}

Finally for $n$ even and $m$ odd we have
\begin{eqnarray}\label{eq:sum_zk_n_even_m_odd}
S_p^{(m)} = S_p^{(1)} = -1~.
\end{eqnarray}

\end{lemma}

\begin{proof}


Let us derive a closed form expression for the cumulative sum $S^{(m)}_{q}$, for any $m,q \in \{1,\ldots,p\}$
\begin{align}\label{eqn:sum-zk-m}
\BA{lll}
S^{(m)}_{q} = \sum_{k=1}^{q}{z_k^{(m)}} &=& \Re\Big(\sum_{k=1}^{q}{e^{\frac{2 i \pi k m}{n}}}\Big)\\
&=& \Re\Big(e^{2 i \pi m/n} \frac{ 1- e^{2 i \pi q m/n} }{ 1- e^{2 i \pi m/n}}\Big)\\
&=& \cos\big(\pi (q+1) m / n\big) \frac{\sin(\pi q m /n)}{\sin(\pi m/n)}~.
\EA
\end{align}


Let us prove equation~\eqref{eqn:sum-total-zk} with the latter expression for $q=p$.
Given that $n=2p+1 = 2 (p+1/2)$, we have,
\begin{eqnarray*}
\frac{\pi (p+1) m }{n} = \frac{\pi (p+1/2 + 1/2) m }{2(p+1/2)} = \frac{\pi m }{2} + \frac{\pi m }{2n}, \\
\frac{\pi p m }{n} = \frac{\pi (p+1/2 - 1/2) m }{2(p+1/2)} = \frac{\pi m }{2} - \frac{\pi m }{2n}~.
\end{eqnarray*}
Now, by trigonometric formulas, we have,
\begin{eqnarray*}
    \cos{\left(\frac{\pi m }{2} + x \right)}=
    \begin{cases}
      (-1)^{m/2}\cos{(x)}, & \text{if}\ \text{$m$ is even} \\
      (-1)^{(m+1)/2} \sin{(x)}, & \text{if}\ \text{$m$ is odd}
    \end{cases}
\end{eqnarray*}
\begin{eqnarray*}
    \sin{\left(\frac{\pi m }{2} - x \right)}=
    \begin{cases}
      (-1)^{(1+m/2)}\sin{(x)}, & \text{if}\ \text{$m$ is even} \\
      (-1)^{(m-1)/2} \cos{(x)}, & \text{if}\ \text{$m$ is odd}
    \end{cases}
\end{eqnarray*}
It follows that, for any $m$,
\begin{eqnarray*}
    \cos{\left(\frac{\pi m }{2} + x \right)} \sin{\left(\frac{\pi m }{2} - x \right)} = - \cos{(x)} \sin{(x)} = - \frac{1}{2} \sin{(2x)}
\end{eqnarray*}
Finally, with $x=\pi m /(2n)$, this formula simplifies the numerator appearing in equation~\eqref{eqn:sum-zk-m} and yields the result in equation~\eqref{eqn:sum-total-zk}.


Let us now prove equation~\eqref{eqn:sum-zk-sym-ineq} with a similar derivation.
Let $f(q) \triangleq \cos\big(\pi (q+1)  / n\big) \sin(\pi q  /n)$, defined for any real $q \in [1,p/2]$.
We wish to prove $f(p-q) \geq f(q)$ for any integer $q \in \{1,\ldots,\lfloor p/2 \rfloor \}$.
Using $n= 2 (p+1/2)$, we have,
\begin{eqnarray*}
\frac{\pi (p-q+1) }{n} = \frac{\pi (p+1/2 - (q- 1/2))  }{2(p+1/2)} = \frac{\pi }{2} - \frac{\pi (q - 1/2)  }{n}, \\
\frac{\pi (p -q) }{n} = \frac{\pi (p+1/2 - (q+1/2))  }{2(p+1/2)} = \frac{\pi }{2} - \frac{\pi (q+1/2) }{n}~.
\end{eqnarray*}
Using $\cos{(\pi/2 - x )} = \sin{(x)}$ and $\sin{(\pi/2 - x )} = \cos{(x)}$, we thus have,
\begin{align}
    f(p-q) = \cos\big(\pi (q+1/2)  / n\big) \sin(\pi (q-1/2) /n) = f(q-1/2)
\end{align}
To conclude, let us observe that $f(q)$ is non-increasing on  $[1,p/2]$.
Informally, the terms $\{z^{1}_k\}_{1\leq k \leq q}$ appearing in the partial sums $S^{(1)}_{q}$ are all non-negative for $q \leq p/2$.
Formally, remark that the derivative of $f$, $df/dq (q) = (\pi/n) \cos{\left( \pi (2q + 1)/n \right)}$ is non-negative for $ q \in [1,p/2]$.
Hence, for $q \leq p/2$, $f(q-1/2) \geq f(q)$, which ends the proof of equation~\eqref{eqn:sum-zk-sym-ineq}.

To get the first equality of \eqref{eq:S_m_p_1}, from the exact form in \eqref{eqn:sum-zk-m}, we have ($n=2p$)
\begin{eqnarray*}
S_{p-1}^{(1)} = \cos(\pi p/(2p)) \frac{\sin(\pi (p-1)/n)}{\sin(\pi/n)} = 0~.
\end{eqnarray*}
For the second equality in \eqref{eq:S_m_p_1}, we have ($m=2q$):
\begin{eqnarray*}
S_{p-1}^{m} &=& \cos(\pi q) \frac{\sin(\pi q -\pi m/n)}{\sin(\pi m /n)} = (-1)^q \frac{-(-1)^q\sin(\pi m/n)}{\sin(\pi m /n)} = -1~.
\end{eqnarray*}

Finally to get \eqref{eq:sum_zk_n_even_m_odd}, let us write ($n=2p$ and $m$ odd):
\begin{eqnarray*}
S_p^{(m)} &=& (-1)^{m+1}\frac{\cos(\pi (p+1)m/n)}{\sin(\pi m/n)} = (-1)^{m+1}\frac{\cos(\pi m/2 + \pi m/n)}{\sin(\pi m/n)}\\
&=& (-1)^m \sin(\pi m/2) = -1~.
\end{eqnarray*}
\end{proof}

The following lemma gives an important property of the partial sum of the $z_k^{(m)}$ that is useful when combined with proposition \ref{prop:partial_sum_trick}.

\begin{lemma}\label{lemma:cos_partial_sum}
Denote by $z_k^{(m)}=\cos({2\pi k m}/{n})$. Consider first $n = 2p$ and $m$ even. For $m= 1,\ldots, p$ and $q=1,\ldots,p-2$
\begin{eqnarray}\label{eq:partial_sum_domination_sub_case}
S_q^{(1)}=\sum_{k=1}^{q}{z_k^{(1)}} \geq \sum_{k=1}^{q}{z_k^{(m)}}=S_q^{(m)}~.
\end{eqnarray}
Otherwise we have for every $(m,q)\in\{1,\ldots,p\}^2$
\begin{eqnarray}\label{eq:partial_sum_domination}
S_q^{(1)} > S_q^{(m)}~,
\end{eqnarray}
with equality when $q= p$.
\end{lemma}

\begin{proof}
\textbf{Case $m$ and $n$ coprime.} Values of $\big(z_k^{(m)}\big)_{k=1,\ldots,p}$ are all distinct. Indeed $z_k^{(m)}= z_{k^\prime}^{(m)}$ implies that $n$ divides $k+k^{\prime}$ or $k-k^{\prime}$. It is impossible (the range of $k+k^{\prime}$ is $[2,2p]$) unless $k=k^{\prime}$.

\textbf{Case $m$ and $n$ not coprime.} $m = d m^{\prime}$ and $n = d n^{\prime}$, with $d\geq 3$. In that situation we need to distinguish according to the parity of $n$. 

\textbf{Case $n=2p+1$.} 
Let's first remark that $\big(z_k^{(1)}\big)_{k=1,\ldots,p}$ takes all values but two ($-1$ and $1$) of the cosinus of multiple of the angle $\frac{2\pi}{n}$, e.g. $\big(z_k^{(1)}\big)_{k=1,\ldots,p}\subset \mathcal{Z}_n$. Also $(z_k^{(1)})_{k=1,\ldots,p}$ is non-increasing. 

Let's prove \eqref{eq:partial_sum_domination} by distinguishing between the various values of $q$.

\begin{itemize}
    \item Consider $q=p-(n^{\prime}-1),\ldots,p$. From
    \eqref{eqn:sum-total-zk} in lemma \eqref{lemma:cos_partial_sum}, we have $S_p^{(1)}=S_p^{(m)}$. The $\big(z_k^{(1)}\big)_k$ are ordered in non-increasing order and the $\big(z_k^{(m)}\big)_{k=p-n^{\prime}+1,\ldots,p}$ take value in $\mathcal{Z}_n\cup\{1\}$ without repetition (it would requires $k\pm k^\prime \sim 0 ~[n^\prime]$). Also the partial sum of $z^{(1)}_k$ starting from the ending point $p$ are lower than any other sequence taking the same or greater value without repetition. Because $1$ is largest than any possible value in $\mathcal{Z}_n$, we hence have
    \begin{eqnarray}\label{eq:partial_from_ending}
    \sum_{k=q}^{p}{z^{(1)}_k} \leq \sum_{k=q}^{p}{z^{(m)}_k}~\text{for any } q=p-(n^{\prime}-1),\ldots,p~.
    \end{eqnarray}
     Since $S_q^{(m)}= S_p^{(m)} - \sum_{k= q+1}^{p}{z^{(m)}_k}$, \eqref{eq:partial_from_ending} implies \eqref{eq:partial_sum_domination} for that particular set of $q$. 

    \item For $q=1,\ldots,n^{\prime}-1$ it is the same type of argument. Indeed the $(z_k^{(1)})_k$ takes the highest values in $\mathcal{Z}_n$ in decreasing order, while $(z_k^{(m)})_k$ takes also its value in $\mathcal{Z}_n$ (because $z_q^{(m)}\neq 1$). This concludes \eqref{eq:partial_sum_domination}.

    Note that when $n^\prime \geq \frac{p+1}{2}$, \eqref{eq:partial_sum_domination} is then true for all $q$. In the sequel, let's then assume that this is not the case, e.g. $n^\prime < \frac{p+1}{2}$. 
    
    \item For $q= n^{\prime}-1,\ldots,\big\lfloor \frac{p}{2}\big\rfloor$, the $z_q^{(1)}$ are non-negative. Hence $S_q^{(1)}$ is non-decreasing and lower bounded by $S_{n^{\prime}-1}^{(1)}$.
    Also because $S^{(m)}_{n^\prime}=0$ and $S^{(1)}_{n^\prime-1}\geq S^{(m)}_{k}$ for $k=1,\ldots,n^\prime$, it is true that for all $q$ in the considered set, $S_q^{(m)}$ is upper-bounded by $S_{n^{\prime}-1}^{(1)}$. All in all it shows \eqref{eq:partial_sum_domination} for these values of $q$.
    
    \item For $q= \big\lfloor \frac{p}{2}\big\rfloor+1 , \ldots , p- n^{\prime}$, we apply \eqref{eqn:sum-zk-sym-ineq} with $q = n^{\prime}$ (and indeed $n^{\prime}\leq\frac{p}{2}$) to get $S_{p-n^{\prime}}^{(1)} \geq S^{(1)}_{n^{\prime}}$. 
    Because $S_q^{(m)}$ is upper-bounded by $S_{n^{\prime}-1}^{(1)}$, it follows that $S_{p-n^{\prime}}^{(1)} \geq S_{q}^{(m)}$. Finally since $(S_{q}^{(1)})$ is non-increasing for the considered sub-sequence of $q$, \eqref{eq:partial_sum_domination} is true.

\end{itemize}

\textbf{Case $n=2p$.}  Here $\big(z_k^{(1)}\big)_{k=1,\ldots,p}$ takes unique values in $\mathcal{Z}_n\cup\{-1\}$. We also need to distinguish according to the parity of $m$.

\begin{itemize}
    \item $\big(z_k^{(m)}\big)_{k=1,\ldots,n^{\prime}-1}$ takes also unique value in $\mathcal{Z}_n$. We similarly get \eqref{eq:partial_sum_domination} for $q=1,\ldots,n^{\prime}-1$, and for $q=n^{\prime}$ because $S_{n^{\prime}}^{(m)}=0$.
    
    \item Consider $m$ odd, from \eqref{eq:sum_zk_n_even_m_odd}, $S^{(m)}_p=S^{(1)}_p=-1$ so that we can do the same reasoning as with $n$ odd to prove \eqref{eq:partial_sum_domination} for $q=p-n^{\prime}+1,\ldots,p$ and $q=1,\ldots,n^{\prime}$. The remaining follows from the symmetry property \eqref{eq:sum_zk_n_m_even} of the sequence $(S_q^{(1)})_q$ in Lemma \ref{lemma:equality_ending_partial_sum}.
    
    \item $m$ and $n$ even, we have that $S_{p-1}^{(1)}=0$ and $S_{p-1}^{(m)}=-1$ so that
    \begin{eqnarray*}
    S_{p-1}^{(1)} \geq S_{p-1}^{(m)}+1~.
    \end{eqnarray*}
    $S_q^{(1)}\geq S_q^{(m)}$ for $q<p-1$ follows with same techniques as before.
\end{itemize}

\end{proof}

\subsection{Properties on R-Toeplitz circular matrix.}

This proposition is a technical method that will be helpful at proving that the eigenvalues of a R-circular Toeplitz matrix are such that $\nu_1>\nu_m$.
\begin{proposition}\label{prop:partial_sum_trick}
Suppose than for any $k=1,\ldots,q:$
\begin{eqnarray*}
W_k\triangleq\sum_{i=1}^{k}{w_i} \geq \sum_{i=1}^{k}{\tilde{w}_i}\triangleq\tilde{W}_k~,
\end{eqnarray*}
with $(w_i)$ and $(\tilde{w}_i)$ two sequences of reals. Then, if $(b_k)_{k}$ is non increasing and non negative, we have
\begin{eqnarray}\label{eq:order_eigen_values}
\sum_{k=1}^{q}{b_k w_k} \geq \sum_{k=1}^{q}{b_k \tilde{w}_k}~.
\end{eqnarray}
\end{proposition}
\begin{proof}
We have
\begin{eqnarray*}
\sum_{k=1}^{q}{b_k w_k} &=& \sum_{k=1}^{q}{b_k (W_k-W_{k-1})}\\
&=& \underbrace{b_q}_{\geq 0} W_q + \sum_{k=1}^{q-1}{\underbrace{(b_k-b_{k+1})}_{\geq 0}W_k}\\
&\geq & b_q \tilde{W}_q + \sum_{k=1}^{q-1}{(b_k - b_{k+1})\tilde{W}_k} = \sum_{k=1}^{q}{b_k\tilde{W}_k}~.
\end{eqnarray*}
\end{proof}

As soon as there exists $k_0\in\{1,\ldots,q\}$ such that
\begin{eqnarray*}
\sum_{i=1}^{k_0}{w_i} > \sum_{i=1}^{k_0}{\tilde{w}_i} ~,
\end{eqnarray*}
then \eqref{eq:order_eigen_values} holds strictly.

The following proposition gives the usual derivations of eigenvalues in the R-circular Toeplitz case.
\begin{proposition}\label{prop:eigen_value_vector_circular_Toeplitz_full}
Consider $A$, a circular-R Toeplitz matrix of size $n$.  

For $n=2p+1$
\begin{eqnarray}
\nu_m & \triangleq & b_0 + 2\sum_{k=1}^{p}{b_k \cos \left(\frac{2 \pi k m}{n}\right)}~.
\end{eqnarray}
For $m = 1,\ldots,p$ each $\nu_m$ are eigenvalues of $A$ with multiplicity 2 and associated eigenvectors
\begin{align}\label{eqn:eigvec-circ}
\BA{ll}
\vspace{.1cm}
y^{m, \text{cos}} = & \frac{1}{\sqrt{n}} \left(1, \cos \left( 2 \pi m / n \right), \ldots, \cos \left( 2 \pi m (n-1) / n \right) \right)\\
\vspace{.1cm}
y^{m, \text{sin}} = & \frac{1}{\sqrt{n}} \left(1, \sin \left( 2 \pi m / n \right), \ldots, \sin \left( 2 \pi m (n-1) / n \right) \right)~.
\EA
\end{align}

For $n = 2p$
\textbf{\begin{align}\label{eqn:eigval-circ-even}
\BA{lll}
\nu_m  & \triangleq & b_0 + 2\sum_{k=1}^{p-1}{b_k \cos \left(\frac{2 \pi k m}{n}\right)} + b_p \cos \left( \pi m \right)~,
\vspace{.1cm}
\EA
\end{align}}
where $\nu_{0}$ is still singular, with $y^{(0)} = \frac{1}{\sqrt{n}} \left(1 ,\ldots, 1\right)~$. 
$\nu_{p}$ also is, with  $y^{(p)} = \frac{1}{\sqrt{n}} \left(+1,-1, \ldots, +1,-1\right)~$, and there are $p-1$ double eigenvalues, for $m=1,\ldots,p-1$, each associated to the two eigenvectors given in equation~\eqref{eqn:eigvec-circ}.

\end{proposition}

\begin{proof}
Let us compute the spectrum of a circular-R, symmetric, circulant Toeplitz matrix. From \citet{gray2006toeplitz}, the eigenvalues are
\begin{align}
\nu_m = \sum_{k=0}^{n-1}{b_k \rho_m^k}~,
\end{align}
with $\rho_m =\exp(\frac{2 i \pi m}{n})$, and the corresponding eigenvectors are,
\begin{align}
y^{(m)} = \frac{1}{\sqrt{n}} \left(1, e^{-2i\pi m/n},\ldots, e^{-2i\pi m (n-1)/n}\right)~,
\end{align}
for $m=0, \ldots, n-1$.

\textbf{Case $n$ is odd, with $n=2p+1$.}
Using the symmetry assumption $b_k = b_{n-k}$, and the fact that $\rho_m^{n-k} = \rho_m^{n} \rho_m^{-k} = \rho_m^{-k}$, it results in real eigenvalues,
\begin{align}\label{eq:derivation_eigen_values}
\BA{lll}
\vspace{.1cm}
\nu_m & = & b_0 + \sum_{k=1}^{p}{b_k \rho_m^k} + \sum_{k=p+1}^{n-1}{b_k \rho_m^k}\\
\vspace{.1cm}
& = & b_0 + \sum_{k=1}^{p}{b_k \rho_m^k} + \sum_{k=1}^{p}{b_{n-k} \rho_m^{n-k}}\\
\vspace{.1cm}
& = & b_0 + \sum_{k=1}^{p}{b_k (\rho_m^k+\rho_m^{-k})}\\
\vspace{.1cm}
& = & b_0 + 2\sum_{k=1}^{p}{b_k \cos \left(\frac{2 \pi k m}{n}\right)}~.
\EA
\end{align}
Observe also that $\nu_{n-m} = \nu_{m}$, for $m = 1, \ldots, n-1$, resulting in $p+1$ real distinct eigenvalues. $\nu_{0}$ is singular, whereas for $m = 1, \ldots, p$, $\nu_{m}$ has multiplicity $2$, with eigenvectors $y^m$ and $y^{n-m}$.
This leads to the two following real eigenvectors, $y^{m, \text{cos}} = 1/2 (y^m + y^{n-m})$ and $y^{m, \text{sin}}=1/(2i)(y^m - y^{n-m})$
\begin{align}
\BA{ll}
\vspace{.1cm}
y^{m, \text{cos}} = & \frac{1}{\sqrt{n}} \left(1, \cos \left( 2 \pi m / n \right), \ldots, \cos \left( 2 \pi m (n-1) / n \right) \right)\\
\vspace{.1cm}
y^{m, \text{sin}} = & \frac{1}{\sqrt{n}} \left(1, \sin \left( 2 \pi m / n \right), \ldots, \sin \left( 2 \pi m (n-1) / n \right) \right)\\
\EA
\end{align}

\textbf{Case $n$ is even, with $n=2p$.}
A derivation similar to \eqref{eq:derivation_eigen_values} yields,
\begin{align}
\BA{lll}
\nu_m  & = & b_0 + 2\sum_{k=1}^{p-1}{b_k \cos \left(\frac{2 \pi k m}{n}\right)} + b_p \cos \left( \pi m \right)
\vspace{.1cm}
\EA
\end{align}
$\nu_{0}$ is still singular, with $y^{(0)} = \frac{1}{\sqrt{n}} \left(1 ,\ldots, 1\right)~$,
$\nu_{p}$ also is, with  $y^{(p)} = \frac{1}{\sqrt{n}} \left(+1,-1, \ldots, +1,-1\right)~$, and there are $p-1$ double eigenvalues, for $m=1,\ldots,p-1$, each associated to the two eigenvectors given in equation~\eqref{eqn:eigvec-circ}.

\end{proof}

The following proposition is a crucial property of the eigenvalues of a circular Toeplitz matrix. It later ensures that when choosing the second eigenvalues of the laplacian, it will corresponds to the eigenvectors with the lowest period. It is paramount to prove that the latent ordering of the data can be recovered from the curve-like embedding.

\begin{proposition}\label{prop:eigen_values_order}
A circular-R, circulant Toeplitz matrix has eigenvalues $(\nu_m)_{m=0,\ldots,p}$ such that $\nu_1 \geq \nu_m$ for all $m=2,\ldots,p$ with $n=2p$ or $n = 2p+1$.
\end{proposition}

\begin{proof}
Since the shape of the eigenvalues changes with the parity of $n$, let's again distinguish the cases.

For $n$ odd, $\nu_1\geq \nu_m$ is equivalent to showing
\begin{eqnarray}
\sum_{k=1}^{p}{b_k \cos(2\pi k /n)} \geq \sum_{k=1}^{p}{b_k \cos(2\pi k m /n)}~.
\end{eqnarray}
It is true by combining proposition \ref{prop:partial_sum_trick} with lemma \ref{lemma:cos_partial_sum}. The same follows for $n$ even and $m$ odd.


Consider $n$ and $m$ even. We now need to prove that
\begin{eqnarray}\label{eq:eigen_value_even}
2\sum_{k=1}^{p-1}{b_k \cos \left(\frac{2 \pi k }{n}\right)}-b_p\geq 2\sum_{k=1}^{p-1}{b_k \cos \left(\frac{2 \pi k m}{n}\right)}+b_p~.
\end{eqnarray}

From lemma \ref{lemma:cos_partial_sum}, we have that
\begin{eqnarray}
\sum_{k=1}^{q}{z_k^{(1)}} &\geq& \sum_{k=1}^{q}{z_k^{(m)}}~\text{for }q=1,\ldots,p-2\\
\sum_{k=1}^{p-1}{z_k^{(1)}} &\geq& \sum_{k=1}^{p-1}{z_k^{(m)}} +1~.
\end{eqnarray}
Applying proposition \ref{prop:partial_sum_trick} with $w_k = z_k^{(1)}$ and $\tilde{w}_k=z_k^{(m)}$ for $k\leq p-2$ and $\tilde{w}_{p-1} = z_{p-1}^{(m)}+1$, we get
\begin{eqnarray}
\sum_{k=1}^{p-1}{z_k^{(1)}b_k} \geq \sum_{k=1}^{p-1}{b_kz_k^{(m)}} +b_{p-1}\\
2\sum_{k=1}^{p-1}{z_k^{(1)}b_k} \geq 2\sum_{k=1}^{p-1}{b_kz_k^{(m)}} +2b_{p}~.
\end{eqnarray}
The last inequality results from the monotonicity of $(b_k)$ and is equivalent to \eqref{eq:eigen_value_even}. It concludes the proof.
\end{proof}

\subsection{Recovering exactly the order.}
Here we provide the proof for Theorem \ref{th:without_noise}.
\begin{theorem}\label{th:proof_without_noise}
Consider the seriation problem from an observed matrix $\Pi S\Pi^T$, where $S$ is a R-circular Toeplitz matrix. Denote by $L$ the associated graph Laplacian. Then the two dimensional laplacian spectral embedding (\eqref{eqn:lapl-embed} with d=2) of the items lies ordered and equally spaced on a circle.
\end{theorem}

\begin{proof}
Denote $A= \Pi S\Pi^T$. The unnormalized Laplacian of $A$ is $L \triangleq \text{diag}(A 1)-A$. The eigenspace associated to its second smallest eigenvalue corresponds to that of $\mu_1$ in $A$. $A$ and $S$ share the same spectrum. Hence the eigenspace of $\mu_1$ in $A$ is composed of the two vectors $\Pi y^{1,sin}$ and $\Pi y^{1,cos}$. 

Denote by $(p_i)_{i=1,\ldots,n}\in\mathbb{R}^2$ the \kLE{2}. Each point is parametrized by
\begin{eqnarray}
p_i = (\cos(2\pi \sigma(i)/n),\sin(2\pi \sigma(i)/n))~,
\end{eqnarray}
where $\sigma$ is the permutation represented matricially by $\Pi$.
\end{proof}

%% file: sections/6_appendix_perturbation_analysis.tex
\section{Perturbation analysis}\label{sec:perturbation_analysis}
The purpose of the following is to provide guarantees of robustness to the noise with respect to quantities that we will not try to explicit. Some in depths perturbation analysis exists in similar but simpler settings \citep{Fogel}. In particular, linking performance of the algorithm while controlling the perturbed embedding is much more challenging than with a one dimensional embedding. 

We have performed graph Laplacian re-normalization to make the initial similarity matrix closer to a Toeplitz matrix. Although we cannot hope to obtain exact Toeplitz Matrix. Hence perturbation analysis provide a tool to recollect approximate Toeplitz matrix with guarantees to recover the ordering.

\subsection{Davis-Kahan}

We first characterize how much each point of the new embedding deviate from its corresponding point in the rotated initial set of points. Straightforward application of Davis-Kahan provides a bound on the Frobenius norm that does not grant directly for individual information on the deviation.

\begin{proposition}[Davis-Kahan]\label{prop:davis_kahan_circular}

Consider $L$ a graph Laplacian of a R-symmetric-circular Toeplitz matrix $A$. We add a symmetric perturbation matrix $H$ and denote by $\Tilde{A} = A + H$ and $\tilde{L}$ the new similarity matrix and graph Laplacian respectively. Denote by $(p_i)_{i=1,\ldots,n}$ and $(\tilde{p}_i)_{i=1,\ldots,n}$ the \kLE{2} coming from $L$ and $\tilde{L}$ respectively. Then there exists a cyclic permutation $\tau$ of $\{1,\ldots,n\}$ such that
\begin{eqnarray}
\sup_{i=1,\ldots,n} ||p_{\tau(i)} - \tilde{p}_i||_2 \leq   \frac{2^{3/2}\min(\sqrt{2}||L_H||_2,||L_H||_F)}{\min(|\lambda_1|,|\lambda_2-\lambda_1|)}~,
\end{eqnarray}\label{eq:perturbation_analysis}
where $\lambda_1<\lambda_2$ are the first non-zeros eigenvalues of $L$.
\end{proposition}

\begin{proof}
For a matrix $V\in\mathbb{R}^{n\times d}$, denote by
\begin{eqnarray*}
\big\vert\big\vert V \big\vert\big\vert_{2,\infty} = \sup_{i=1,\ldots,n} \big\vert\big\vert V_i \big\vert\big\vert_2~,
\end{eqnarray*}
where $V_i$ are the columns of $V$.
Because in $\mathbb{R}^n$ we have $||\cdot||_{\infty} \leq ||\cdot||_2$, it follows that
\begin{eqnarray*}
\big\vert\big\vert V \big\vert\big\vert_{2,\infty} &\leq& \big\vert\big\vert   \big(||V_i||\big)_{i=1,\ldots,n} \big\vert\big\vert_2 = \sqrt{\sum_{i=1}^{n}{||V_i||_2^2}}\\
&\leq& \big\vert\big\vert V \big\vert\big\vert_{F}~.
\end{eqnarray*}

 We apply \citep[Theorem 2]{yu2014useful} to our perturbed matrix, a simpler version of classical davis-Kahan theorem \citep{davis1970rotation}.
 
 Let's denote by $(\lambda_1,\lambda_2)$ the first non-zeros eigenvalues of $L$ and by $V$ its associated 2-dimensional eigenspace. Similarly denote by $\tilde{V}$ the 2-dimensional eigenspace associated to the first non-zeros eigenvalues of $\tilde{L}$. There exists a rotation matrix $O\in SO_2(\mathbb{R})$ such that
\begin{eqnarray}
||\tilde{V}-VO||_F \leq \frac{2^{3/2}\min(\sqrt{2}||L_H||_2,||L_H||_F)}{\min(|\lambda_1|,|\lambda_2-\lambda_1|)} ~.
\end{eqnarray}

In particular we have
\begin{eqnarray*}
\big\vert\big\vert \tilde{V} - V O \big\vert\big\vert_{2,\infty} &\leq& \big\vert\big\vert \tilde{V} - V O \big\vert\big\vert_{F} \\
\big\vert\big\vert \tilde{V} - V O \big\vert\big\vert_{2,\infty} &\leq& \frac{2^{3/2}\min(\sqrt{2}||L_H||_2,||L_H||_F)}{\min(|\lambda_1|,|\lambda_2-\lambda_1|)}
\end{eqnarray*}


Finally because $A$ is a R-symmetric-circular Toeplitz, from Theorem \ref{th:without_noise}, the row of $V$ are $n$ ordered points uniformly sampled on the unit circle. Because applying a rotation is equivalent to translating the angle of these points on the circle. It follows that there exists a cyclic permutation $\tau$ such that
\begin{eqnarray*}
\sup_{i=1,\ldots,n} ||p_i - \tilde{p}_{\tau(i)}||_2 \leq   \frac{2^{3/2}\min(\sqrt{2}||L_H||_2,||L_H||_F)}{\min(|\lambda_1|,|\lambda_2-\lambda_1|)}~,
\end{eqnarray*}

\end{proof}

\subsection{Exact recovery with noise for Algorithm \ref{alg:circular-2d-ordering}}
When all the points remain in a sufficiently small ball around the circle, Algorithm \ref{alg:circular-2d-ordering} can exactly find the ordering. Let's first start with a geometrical lemma quantifying the radius of the ball around each $(\cos(\theta_k),\sin(\theta_k))$ so that they do not intersect.

\begin{lemma}\label{lemma:variation_angle_arcos}
For $\xx\in\mathbb{R}^2$ and $\theta_k = 2\pi k/n$ for $k\in\mathbb{N}$ such that
\begin{eqnarray}\label{eq:safe_ball}
||\xx - (\cos(\theta_k),\sin(\theta_k))||_2 \leq \sin(\pi/n)~,
\end{eqnarray}
we have
\begin{eqnarray*}
|\theta_x - \theta_k|\leq \pi/n~,
\end{eqnarray*}
where $\theta_x = \tan^{-1}(\xx_1/\xx_2) + 1 [\xx_1<0]\pi$.
\end{lemma}
\begin{proof}
Let $\xx$ that satisfies \eqref{eq:safe_ball}. Let's assume without loss of generality that $\theta_k=0$ and $\theta_x\geq 0$. Assume also that $\xx = \boldsymbol{e}_1+\sin(\pi/n) \uu_x$ where $\uu$ is a unitary vector. A $\xx$ for which $\theta_x$ is maximum over these constrained is such that $\uu_x$ and $\xx$ are orthonormal.

Parametrize $\uu_x = (\cos(\gamma),\sin(\gamma))$, because $\uu_x$ and $\xx$ are orthonormal, we have $\cos(\gamma) = \sin(-\pi/n)$. Finally since $\theta_x\geq 0$, it follows that $\gamma = \pi/2 +\pi/n$ and hence with elementary geometrical arguments $\theta_x = \pi/n$.

\end{proof}

\begin{proposition}[Exact circular recovery under noise in Algorithm \ref{alg:circular-2d-ordering}]
Consider a matrix $\tilde{A} = \Pi^T A \Pi + H$ with $A$ a $R-$circular Toeplitz ($\Pi$ is the matrix associated to the permutation $\sigma$) and $H$ a symmetric matrix such that
\begin{eqnarray*}
\min(\sqrt{2}||L_H||_2,||L_H||_F) \leq 2^{-3/2} \sin(\pi/n) \min(|\lambda_1|,|\lambda_2-\lambda_1|)~,
\end{eqnarray*}
where $\lambda_1<\lambda_2$ are the first non-zeros eigenvalues of the graph Laplacian of $\Pi^T A \Pi$. Denote by $\hat{\sigma}$ the output of Algorithm \ref{alg:circular-2d-ordering} when having $\tilde{A}$ as input. Then there exists a cyclic permutation $\tau$ such that
\begin{eqnarray}
\hat{\sigma} = \sigma^{-1}\circ\tau^{-1}~.
\end{eqnarray}
\end{proposition}

\begin{proof}
We have
\begin{eqnarray*}
\Pi^T\tilde{A}\Pi = A + \Pi^T H\Pi~.
\end{eqnarray*}
$L$ is the graph Laplacian associated to $A$ and $\tilde{L}$, the one associated to $\tilde{A}$. Denote by $(p_i)_{i=1,\ldots,n}$ and $(\tilde{p}_i)_{i=1,\ldots,n}$ the \kLE{2} coming from $L$ and $\tilde{L}$ respectively. $(\tilde{p}_{\sigma^{-1}(i)})_{i=1,\ldots,n}$ is the \kLE{2} coming from the graph Laplacian of $\Pi^T\tilde{A}\Pi$.

Applying Proposition \ref{prop:davis_kahan_circular} with $\Pi^T\tilde{A}\Pi$, there exists a cyclic permutation such that
\begin{eqnarray*}
\sup_{i=1,\ldots,n} ||\tilde{p}_{\sigma^{-1}(i)} - p_{\tau(i)}||_2 <   \frac{2^{3/2}\min(\sqrt{2}||L_{H^{\pi}}||_2,||L_{H^{\pi}}||_F)}{\min(|\lambda_1|,|\lambda_2-\lambda_1|)}~,
\end{eqnarray*}
with $H^{\pi} = \Pi^T H\Pi$, $\lambda_1<\lambda_2$ the first non zero eigenvalues of $A$.

Graph Laplacian involve the diagonal matrix $D_H$. In particular we have that $D_{H^{\pi}} = \Pi^T D_H \Pi$. For the unnormalized Laplacian, it results in $L_{H^{\pi}} = \Pi^T L_H \Pi$. We hence have
\begin{eqnarray*}
\sup_{i=1,\ldots,n} ||\tilde{p}_{\sigma(i)} - p_{\tau(i)}||_2 & < &  \frac{2^{3/2}\min(\sqrt{2}||L_{H}||_2,||L_{H}||_F)}{\min(|\lambda_1|,|\lambda_2-\lambda_1|)}\\
\sup_{i=1,\ldots,n} ||\tilde{p}_{i} - p_{\tau\circ\sigma^{-1}(i)}||_2 & < & \sin(\pi/n)~.
\end{eqnarray*}

From Theorem \ref{th:without_noise}, $p_i = \cos(2\pi i /n)$ for all $i$. It follows that for any $i$
\begin{eqnarray*}
||\tilde{p}_{i} - \cos(2\pi \tau\circ\sigma(i) /n)||_2 & < & \sin(\pi/n)~.
\end{eqnarray*}

Algorithm \ref{alg:circular-2d-ordering} recovers the ordering by sorting the values of 
\begin{eqnarray*}
\theta_i = \tan^{-1}(\tilde{p}_i^1/\tilde{p}_i^2) + 1 [\tilde{p}_i^1<0]\pi~,
\end{eqnarray*}
where $\tilde{p}_i = (\tilde{p}_i^1, \tilde{p}_i^2)$. Applying Lemma \ref{lemma:variation_angle_arcos}:
\begin{eqnarray*}
|\theta_i - 2\pi (\tau\circ\sigma^{-1})(i) /n| < \pi/n~~\forall i\in\{1,\ldots,n\},
\end{eqnarray*}
so that 
\begin{eqnarray}
\theta_{\sigma^{-1}\circ\tau^{-1}(1)}\leq \cdots \leq \theta_{\sigma^{-1}\circ\tau^{-1}(n)}~.
\end{eqnarray}
Finally $\Hat{\sigma} = \sigma^{-1}\circ\tau^{-1}$.
\end{proof}